\documentclass[pra,twocolumn,notitlepage,
superscriptaddress,aps,amsmath,amssymb,UFT8]{revtex4-2}
\renewcommand{\thetable}{\arabic{table}}
\usepackage[usenames,dvipsnames]{color} %
\usepackage{xcolor}
\usepackage{graphicx}%
\usepackage{bm}%
\usepackage{CJK}
\definecolor{LinkColor}{rgb}{0,0,.5}
\usepackage[colorlinks=true,linkcolor=BrickRed,citecolor=LinkColor,urlcolor=LinkColor]{hyperref}
\usepackage{multirow}
\usepackage{comment}
\usepackage{lineno}
\usepackage{algorithm}
\usepackage{algpseudocode}
\usepackage{amsthm}

\newtheorem{theorem}{Theorem}
\newtheorem{lemma}{Lemma}

\newcommand{\ket}[1]{\left\vert{#1}\right\rangle}
\newcommand{\bra}[1]{\left\langle{#1}\right\vert}

\newcommand{\Gate}[1]{\textsc{#1}}

\newcommand{\zgate}{\Gate{z}}

\newcommand{\xgate}{\Gate{x}}
\newcommand{\czgate}{\Gate{cz}}

\newcommand{\idgate}{\Gate{i}}

\newcommand{\cnotgate}{\Gate{cnot}}
\newcommand{\cxgate}{\Gate{cnot}}

\newcommand*{\Prob}{\mathsf{Pr}}
\newcommand{\popt}{p^{\text{opt}}}
\newcommand{\QAOATTS}{1.46}
\newcommand{\QAOAAATTS}{1.21}

\newcommand{\rev}[1]{#1}

\makeatletter
\renewcommand\NAT@biblabelnum[1]{(#1)}
\renewcommand\NAT@citenum[3]{\ifNAT@swa
(#1)%
   \if\relax#3\relax\else\ (#3)\fi\else (#1)\fi\endgroup}
\makeatother

\begin{document}

\title{Evidence of Scaling Advantage for the Quantum Approximate \\ Optimization Algorithm on a Classically Intractable Problem}
\author{Ruslan Shaydulin}
\thanks{Corresponding author. Email: \mbox{ruslan.shaydulin@jpmchase.com}}
\affiliation{Global Technology Applied Research, JPMorgan Chase, New York, NY 10017, USA}
\author{Changhao Li} %
\affiliation{Global Technology Applied Research, JPMorgan Chase, New York, NY 10017, USA}
\author{Shouvanik Chakrabarti}
\affiliation{Global Technology Applied Research, JPMorgan Chase, New York, NY 10017, USA}
\author{Matthew DeCross}
\affiliation{Quantinuum, Broomfield, CO 80021, USA}
\author{Dylan Herman}
\affiliation{Global Technology Applied Research, JPMorgan Chase, New York, NY 10017, USA}
\author{Niraj~Kumar} %
\affiliation{Global Technology Applied Research, JPMorgan Chase, New York, NY 10017, USA}
\author{Jeffrey Larson} %
\affiliation{Mathematics and Computer Science Division, Argonne National Laboratory, Lemont, IL 60439, USA}
\author{Danylo Lykov}
\affiliation{Global Technology Applied Research, JPMorgan Chase, New York, NY 10017, USA}
\affiliation{Computational Science Division, Argonne National Laboratory, Lemont, IL 60439, USA}
\author{Pierre Minssen}
\affiliation{Global Technology Applied Research, JPMorgan Chase, New York, NY 10017, USA}
\author{Yue Sun}
\affiliation{Global Technology Applied Research, JPMorgan Chase, New York, NY 10017, USA}
\author{Yuri~Alexeev}
\affiliation{Computational Science Division, Argonne National Laboratory, Lemont, IL 60439, USA}
\author{Joan~M.~Dreiling}
\affiliation{Quantinuum, Broomfield, CO 80021, USA}
\author{John~P.~Gaebler}
\affiliation{Quantinuum, Broomfield, CO 80021, USA}
\author{Thomas M. Gatterman}
\affiliation{Quantinuum, Broomfield, CO 80021, USA}
\author{Justin A. Gerber}
\affiliation{Quantinuum, Broomfield, CO 80021, USA}
\author{Kevin Gilmore}
\affiliation{Quantinuum, Broomfield, CO 80021, USA}
\author{Dan Gresh}
\affiliation{Quantinuum, Broomfield, CO 80021, USA}
\author{Nathan~Hewitt}
\affiliation{Quantinuum, Broomfield, CO 80021, USA}
\author{Chandler V. Horst}
\affiliation{Quantinuum, Broomfield, CO 80021, USA}
\author{Shaohan Hu}
\affiliation{Global Technology Applied Research, JPMorgan Chase, New York, NY 10017, USA}
\author{Jacob Johansen}
\affiliation{Quantinuum, Broomfield, CO 80021, USA}
\author{Mitchell Matheny}
\affiliation{Quantinuum, Broomfield, CO 80021, USA}
\author{Tanner Mengle}
\affiliation{Quantinuum, Broomfield, CO 80021, USA}
\author{Michael~Mills}
\affiliation{Quantinuum, Broomfield, CO 80021, USA}
\author{Steven A. Moses}
\affiliation{Quantinuum, Broomfield, CO 80021, USA}
\author{Brian Neyenhuis}
\affiliation{Quantinuum, Broomfield, CO 80021, USA}
\author{Peter Siegfried}
\affiliation{Quantinuum, Broomfield, CO 80021, USA}
\author{Romina Yalovetzky}
\affiliation{Global Technology Applied Research, JPMorgan Chase, New York, NY 10017, USA}
\author{Marco Pistoia}
\affiliation{Global Technology Applied Research, JPMorgan Chase, New York, NY 10017, USA}

\begin{abstract}
The quantum approximate optimization algorithm (QAOA) is a leading candidate algorithm for solving optimization problems on quantum computers. However, the potential of QAOA to tackle classically intractable problems remains unclear. In this paper, we perform an extensive numerical investigation of QAOA on the Low Autocorrelation Binary Sequences (LABS) problem, which is classically intractable even for moderately sized instances. We perform noiseless simulations with up to 40 qubits and observe that the runtime of QAOA with fixed parameters scales better than branch-and-bound solvers, which are the state-of-the-art exact solvers for LABS. The combination of QAOA with quantum minimum-finding gives the best empirical scaling of any algorithm for the LABS problem. We demonstrate experimental progress in executing QAOA for the LABS problem using an algorithm-specific error detection scheme on Quantinuum trapped-ion processors. Our results provide evidence for the utility of QAOA as an algorithmic component that enables quantum speedups.
\end{abstract}

\maketitle

\section*{Introduction}

Quantum computers have been shown to have the potential to speed up the solution of optimization problems. At the same time, only a small number of algorithmic primitives are known that provide broadly applicable speedups.
These include amplitude amplification~\cite{quant-ph/9607014}, quantum walks~\cite{montanaro2018quantum,montanaro2020quantum,2210.03210} and quantum Markov Chain algorithms~\cite{Somma2008,Wocjan2008}, as well as the recently introduced short path
algorithm~\cite{Hastings2018,2212.01513}.

The dearth of provable speedups in quantum optimization motivates the development of heuristics. %
A leading candidate for demonstrating a heuristic speedup in quantum optimization is the quantum approximate optimization algorithm (QAOA)~\cite{Hogg2000,farhi2014quantum}.  QAOA uses two operators applied in alternation $p$ times to prepare a quantum state such that, upon measuring it, a high-quality solution to the problem is obtained with high probability. A pair of such operators is commonly referred to as one QAOA ``layer.'' The state is evolved with a diagonal Hamiltonian encoding the optimization problem by the first operator and with a non-diagonal transverse-field Hamiltonian by the second operator.
In this work, we consider the evolution times to be hyperparameters that are set by using a fixed, predetermined rule, analogously to the choice of a schedule in simulated annealing.

\begin{figure*}
    \centering
    \includegraphics[width=0.95\textwidth]{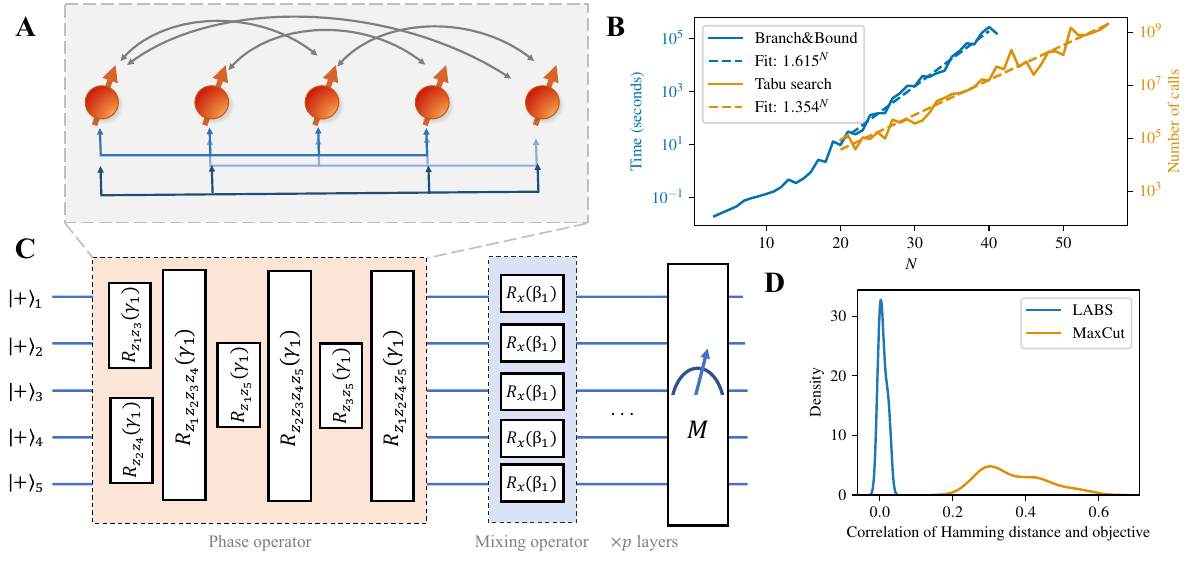}
    \caption{\textbf{Classical and quantum algorithms applied to the LABS problem.} \textbf{A,} Diagram of the LABS problem (with example of $N=5$). The problem involves non-local two-body (black lines) and four-body (blue lines) interactions. \textbf{B,} Time-to-solution (TTS) of classical solvers. For the sizes considered, we observe clear exponential scaling with exponents matching their asymptotic values reported in the literature (see Table~\ref{tab:summary}). \textbf{C,} Diagram of QAOA circuit for a 5-qubit example. Starting from a uniform superposition of the computational basis states, we apply $p$ layers of phase and mixing operators, followed by measurement in the computational basis. \textbf{D,} The distribution over $21\leq N \leq 31$ (for LABS) and $34$ random instances (for MaxCut on random 3-regular graphs with 20 nodes) of Pearson product-moment correlation coefficients relating the Hamming distance of bitstrings from the optimal solution with the objective value of the bitstring. 
    LABS has a much lower correlation between the Hamming distance and objective, indicating that it is much harder than the commonly considered MaxCut problem. 
    }
    \label{fig:Fig1_scheme_summary}
\end{figure*}

While QAOA has been studied extensively~\cite{zhou2020quantum,lipics.tqc.2022.7,2208.06909,Sureshbabu2023}, 
little is known about its potential to provide a scaling advantage over classical solvers. A recent numerical study~\cite{2208.06909} of random 8-SAT with $N\leq 20$ variables has shown that the time-to-solution (TTS) of QAOA %
with fixed parameters and constant depth grows as $1.23^N$. When QAOA is combined with amplitude amplification, the quantum TTS grows as $1.11^N$~\cite{2208.06909}, whereas the best classical heuristic has TTS that grows as $1.25^N$~\cite{2208.06909}. 
Our work is motivated by this preliminary numerical evidence on small instances, which indicates that QAOA may potentially scale better than classical solvers when executed on an idealized quantum computer.

We study the scaling of QAOA TTS with the problem size on the Low Autocorrelation Binary Sequences (LABS) problem~\cite{Boehmer1967,Schroeder1970}, also known as the  Bernasconi model in  statistical physics~\cite{Bernasconi1987,cond-mat/9707104}. The LABS problem has applications in communications engineering, where the low autocorrelation sequences are used for designing radar pulses~\cite{Boehmer1967,Golay1977}. %
To solve LABS, one has to produce a sequence of $N$ bits that minimizes a specific quartic objective. 

We choose LABS to study the scaling of QAOA TTS for the following three reasons. First, the complexity of LABS grows rapidly, with optimal solutions known only for $N\leq 66$ and the best heuristics producing approximate solutions of quality decaying with $N$ for $N \gtrapprox 200$~\cite{OPUS2-git_labs-Boskovic,Packebusch2016}. 
This makes it a promising candidate problem, since only a few hundred qubits are required to tackle classically intractable instances.
Second, the performance of classical solvers for LABS has been benchmarked~\cite{OPUS2-git_labs-Boskovic,Packebusch2016} in terms of the scaling of their TTS with problem size. \rev{Since optimal solutions are only known for $N\leq 66$, the scaling of TTS for all classical solvers is obtained by fitting results for $N\leq 66$.} We reproduce these results and observe that that the scaling of classical solvers at $N\leq 40$ matches the behavior \rev{for $N$ up to $66$} reported in the literature. This provides evidence that the scaling we observe for QAOA at $N\leq 40$ will similarly extrapolate to larger $N$.
Third, LABS has only one instance per problem size $N$. Combined with the hardness of LABS, this makes it possible to reliably study the scaling of QAOA at large problem sizes, where simulating tens or hundreds of random instances would be computationally infeasible. %

We obtain the scaling by performing noiseless exact simulation of QAOA with fixed schedules. Our results are enabled by a custom algorithm-specific GPU simulator~\cite{FUR}, which we execute using up to 1,024 GPUs per simulation on the Polaris supercomputer accessed through the Argonne Leadership Computing Facility. We find that the TTS of QAOA with number of layers $p = 12$ grows as $\QAOATTS^N$, which is improved to $\QAOAAATTS^N$ if combined with quantum minimum-finding. This scaling is better than that of the best classical heuristic, which has a TTS that grows as $1.34^N$.
\rev{We note that we do not propose any new quantum algorithms in this work. Instead, we study} a general quantum optimization heuristic with broad applicability \rev{(namely, QAOA) and make} no specific modifications to adapt it to the LABS problem. %

Our numerical evidence indicates that the proposed quantum algorithm scales better than the best classical heuristic in an idealized setting.
However, we do not claim that QAOA is the best theoretically possible algorithm for the LABS problem.
In particular, it may be possible to quadratically accelerate the best-known classical heuristic (Memetic Tabu~\cite{Gallardo2009}) by applying ideas similar to those used in quantum simulated annealing~\cite{Somma2008,Lemieux2020,Boixo2015}.
Nonetheless, our results highlight the potential of QAOA to act as a useful algorithmic component that enables super-Grover quantum speedups. 

As a first step toward execution of QAOA for the LABS problem, %
we implement QAOA on Quantinuum trapped-ion quantum processors~\cite{Pino2021, Moses2023RaceTrack} on problems with up to $N=18$. We further implement an algorithm-specific error detection scheme inspired by Pauli error detection~\cite{Gonzales2023,Debroy2020} and demonstrate that it can reduce the impact of noise on solution quality by up to $65\%$. %
Our experiments highlight the continuing improvements to quantum computing hardware and %
the potential of algorithm-specific techniques to reduce the overhead of error detection and correction.

\begin{figure*}
    \centering
    \includegraphics[width=\textwidth]{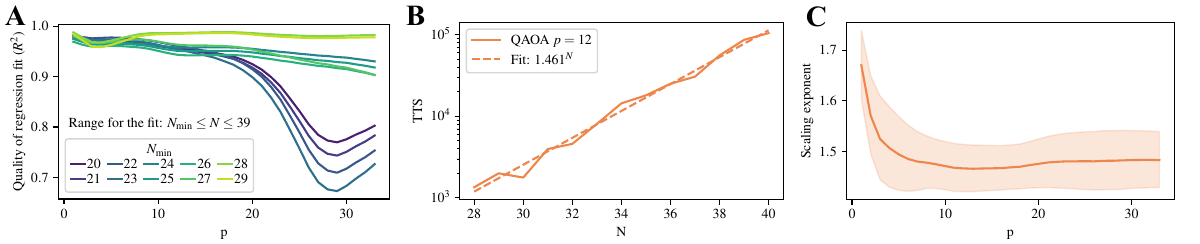}
    \caption{\textbf{QAOA runtime scaling.} \textbf{A,} The quality of the exponential fit for different choices of minimum $N$ to include in the fit. $N\geq 28$ results in a robust fit, the quality of which does not deteriorate with $p$. $N=40$ is omitted as it was only simulated up to $p=22$.
    \textbf{B,} TTS of QAOA at $p=12$. Clear exponential scaling is observed. \textbf{C, } The scaling exponent of QAOA runtime for different QAOA depths $p$. Shaded area shows $95\%$ confidence interval. Increasing $p$ beyond $p\approx 12$ does not lead to better scaling.}
    \label{fig:Fig2_pt_1_QAOA}
\end{figure*}

\section*{Results}
\subsection*{Problem statement}\label{sec:background}
We begin by formally defining the LABS problem, discussing the state of the art of classical LABS solvers, and describing how QAOA is applied to solve the problem.

For a given sequence of spins $s_i\in\{\pm 1\}$, the autocorrelation is given as 
\begin{equation}
    \mathcal{A}_k(\boldsymbol{s}) = \sum_{i=1}^{N-k}s_is_{i+k}.
\end{equation}
The goal of the LABS problem is to find a sequence of spins that minimizes the so-called ``sidelobe'' energy, 
\begin{equation}\label{eq:labs_en}
    \mathcal{E}_{\text{sidelobe}}(\boldsymbol{s}) = \sum_{k=1}^{N-1} \mathcal{A}_k^2(\boldsymbol{s}),
\end{equation}
or, equivalently, maximizes the merit factor
\begin{equation}\label{eq:labs_MF}
    \mathcal{F}(\boldsymbol{s}) = \frac{N^2}{2\mathcal{E}_{\text{sidelobe}}(\boldsymbol{s})}.
\end{equation}
The time-to-solution (TTS) is defined as the time a solver takes to produce this sequence. The energy $\mathcal{E}_{\text{sidelobe}}(\boldsymbol{s})$ is a polynomial containing terms of degree 2 and 4, visualized in Fig.~\ref{fig:Fig1_scheme_summary}A.
\rev{An instance of the LABS problem is unique for each $N$, and}
can be encoded on qubits by the following Hamiltonian:
\begin{equation}\label{eq:labs_hamiltoninan}
\begin{aligned}
    H_C = 2\sum_{i=1}^{N-3}\zgate_i & \sum_{t=1}^{\lfloor \frac{N-i-1}{2}\rfloor}\sum_{k=t+1}^{N-i-t}\zgate_{i+t}\zgate_{i+k}\zgate_{i+k+t} \\
    &+\sum_{i=1}^{N-2}\zgate_i\sum_{k=1}^{\lfloor \frac{N -i}{2}\rfloor}\zgate_{i+2k},
\end{aligned}
\end{equation}
where $\zgate_j$ is a Pauli $\zgate$ operator acting on qubit $j$.

The runtimes of state-of-the-art classical solvers for the LABS problem scale exponentially, with clear exponential scaling present at $N\leq 40$ as shown in Fig.~\ref{fig:Fig1_scheme_summary}B.
The best-known exact solvers are branch-and-bound methods that have a running time that scales as $1.73^N$~\cite{Packebusch2016}. The best-known heuristic for general LABS is tabu search initialized with a memetic algorithm (Memetic Tabu)~\cite{Gallardo2009}, and has a running time that scales as $1.34^N$~\cite{Bokovi2017}. \rev{We provide a survey of classical solvers for the LABS problem in Supplementary Information.}

To see why LABS is harder to solve than other commonly studied problems such as MaxCut, we can examine the correlation between the Hamming distance to the optimal solution and the objective.
The comparison is shown in Fig.~\ref{fig:Fig1_scheme_summary}C. This correlation is one example of problem structure used by both classical and quantum heuristics to solve the problem quickly~\cite{Hogg2000}. The absence of this correlation highlights the hardness of LABS compared with other commonly considered problems such as MaxCut.

As a consequence of the exponential scaling, the LABS problem becomes classically intractable at moderate sizes. %
Specifically, the value of the best-known merit factor decreases notably for high $N$, whereas the asymptotic limit predicts that the merit factor should stay approximately constant. 
This failure of state-of-the-art heuristics has been observed for $N > 200$~\cite{OPUS2-git_labs-Boskovic,Packebusch2016}.
The clear failure of the classical method to obtain high-quality solutions even at small sizes makes LABS an appealing candidate problem for quantum optimization heuristics.

\begin{figure*}
    \centering
    \includegraphics[width=\textwidth]{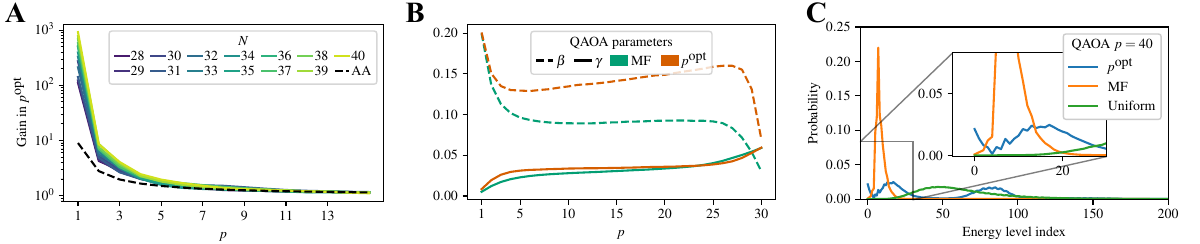}
    \caption{\textbf{QAOA dynamics under different parameter schedules.} \textbf{A,} The gain in success probability $\popt$ from applying step $p$ of QAOA and amplitude amplification (AA). The gain is defined as $\popt_{\text{at step }p} / \popt_{\text{at step }(p-1)}$. The gain at $p=1$ is over the random guess. Only one line is plotted for amplitude amplification since the lines for the values of $N$ considered are visually indistinguishable. For small $p$, a QAOA layer gives orders of magnitude larger gain than a step of AA. \textbf{B,} Fixed QAOA parameters
    for $p=30$ chosen with respect to the QAOA energy $\langle C\rangle_{\text{MF}}$ (``MF'') and probability of sampling the optimal solution (``$\popt$''). 
    Different choice of optimization objective gives different resulting parameters. \textbf{C,} Probability of obtaining a binary string corresponding to a given energy level of the LABS problem (the zeroth energy level is the ground state or optimal solution; lower is better). When parameters are optimized with respect to the expected merit factor (labeled ``MF''), the QAOA output state is concentrated around the mean and fails to obtain a high overlap with the ground state. On the other hand, when parameters are optimized with respect to $\popt$ (labeled ``$\popt$''), the QAOA state has a high overlap with both the ground state and higher energy states. The probability of obtaining the ground state is $27.3$ times greater for QAOA with parameters optimized with respect to $\popt$ at $p=40$.}
    \label{fig:Fig2_pt_2_QAOA}
\end{figure*}

In this work, we tackle the LABS problem using QAOA.
As shown in the circuit diagram Fig.~\ref{fig:Fig1_scheme_summary}D, QAOA solves optimization problems by preparing a parameterized state 
\begin{equation}\label{eq:qaoa_ansatz}
    \vert \bm \beta, \bm \gamma\rangle = \prod_{l=1}^p e^{-i \beta_l \sum_{j=1}^{N}\xgate_j} e^{-i \gamma_l H_C} \vert + \rangle^{\otimes N}, 
\end{equation}
where $\vert + \rangle^{\otimes N}$ is a uniform superposition over computational basis states, $H_C$ is the diagonal Hamiltonian encoding the problem, and $\xgate_j$ is a Pauli $\xgate$ operator acting on qubit $j$. The operator $e^{-i \gamma H_C}$ is commonly referred to as the phase operator and $e^{-i \beta \sum_{j=1}^{N}\xgate_j}$ as the mixing operator. The evolution times $\bm \beta, \bm \gamma$ are hyperparameters chosen to maximize some figure of merit, such as the expected quality of the measurement outcomes or the probability of measuring the optimal solution.
While $\bm \beta, \bm \gamma$ can be optimized independently for each problem size, we consider them to be hyperparameters and use one fixed set of parameters for the LABS problem with a given QAOA depth $p$ regardless of size. The fixed set of parameters is obtained by optimizing $\bm \beta, \bm \gamma$ numerically for a number of small problem sizes and introducing an averaging and rescaling procedure to extrapolate parameters to any problem size (see the Methods section).

When choosing the parameters $\bm \beta, \bm \gamma$ and evaluating the quality of the solution obtained by QAOA, two figures of merit are commonly considered. The first one is the expected merit factor of the sampled binary strings, given by 
\begin{equation}
    \langle C\rangle_{\text{MF}} = \langle \bm \beta, \bm \gamma \vert\frac{N^2}{2 H_C}\vert \bm \beta, \bm \gamma\rangle= \sum_{\boldsymbol{s}\in \{0,1\}^{N}}\Prob(\boldsymbol{s})\mathcal{F}(\boldsymbol{s}).
\end{equation}
We will refer to $\langle C\rangle_{\text{MF}}$ as the ``QAOA energy'' as a shorthand. The second figure of merit  is the probability of sampling the exact optimal solution, denoted by $\popt$ and equal to the sum of squared absolute values of amplitudes of basis states corresponding to exactly optimal solutions. %

\begin{table*}
    \centering
    \begin{tabular}{c|c|c|c|c|c|c|c}
     & \multirow{3}{*}{QAOA$+$QMF} & \multirow{3}{*}{QAOA} & \multicolumn{2}{|c|}{Memetic Tabu} & \multicolumn{3}{|c}{Branch-and-bound}  \\
     & & & \multirow{2}{*}{Reproduced} & \multirow{2}{*}{\cite{Gallardo2009,Bokovi2017}} &  \multicolumn{2}{|c|}{Reproduced} & \multirow{2}{*}{TTO \cite{Packebusch2016}} \\
     & & & & & TTS & TTO & \\
     \hline
     Fit & \textbf{\QAOAAATTS} & \QAOATTS & 1.35 & 1.34 & 1.62 & 1.76 & 1.73\\
     CI & (1.19, 1.23) & (1.42,1.50) & (1.33,1.38) & N/A & (1.57,1.66) & (1.72,1.79) & N/A\\
    \end{tabular}
    \caption{Scaling exponents for quantum and classical algorithms. Confidence intervals (CIs) are $95\%$. The reported asymptotic exponential scaling of classical state-of-the-art solvers is reproduced at $N\leq 40$.  For branch-and-bound, we include both the time to obtain a certificate of optimality (TTO reported in Ref.~\onlinecite{Packebusch2016}) as well as the much shorter time to find an optimal solution (TTS). We observe that QAOA with constant depth of $p=12$ augmented with quantum minimum-finding (``QAOA$+$QMF'') has better time-to-solution scaling than the best known classical heuristics.
    }
    \label{tab:summary}
\end{table*}

In the numerical experiments below, we follow the protocol of Ref.~\onlinecite{2208.06909} and focus on scaling of the QAOA TTS with problem size $N$ as the QAOA depth $p$ is held constant. QAOA TTS is defined as $\frac{1}{\popt}$, i.e. the expected number of measurements required to obtain an optimal solution from the QAOA state.
Ref.~\onlinecite{2208.06909} 
rigorously shows that, for random $k$-SAT, the runtime of constant-depth QAOA grows exponentially with $N$ at any fixed $p$, with the scaling exponent depending on $p$. 
While the nature of the LABS problem makes it difficult to obtain analytical results analogous to Ref.~\onlinecite{2208.06909}, our numerical results also show clear exponential scaling of TTS. We note that, in practice, TTS of QAOA is $\Theta(N^2)\frac{1}{\popt}$, where the $\Theta(N^2)$ prefactor comes from the cost of implementing the LABS phase oracle~\cite{Sanders2020}. However, we do not include it in our analysis because it does not affect the scaling exponent.

\subsection*{Scaling of Quantum Time-to-Solution for LABS problem}\label{sec:QAOA_performance}

We now present the numerical results demonstrating the scaling of TTS of QAOA and QAOA augmented with quantum minimum-finding (``QAOA$+$QMF''). The results are summarized in Table~\ref{tab:summary}. Throughout this section, we present the numerical results obtained using exact noiseless simulations. %
The runtime scaling is obtained by evaluating QAOA once with fixed parameters $\bm\beta, \bm\gamma$ (i.e., with no overhead of parameter optimization) and computing the value $\popt$ with high precision. %
We discuss the parameter setting procedure and the details of simulation in the Methods section.

We are interested in the scaling of the runtime of QAOA for large problem sizes $N$. 
An important question to address is the choice of the smallest $N$ to include in the scaling analysis, since the algorithm's behavior at small sizes may not be predictive of its behavior at large sizes.
Note that the largest $N$ we include is limited by the capability of the classical simulator. We use the quality of the fit as the criterion for the choice of the cutoff on $N$. Figure~\ref{fig:Fig2_pt_1_QAOA}A shows that if we set the cutoff at $N\geq 28$, we obtain a robust high-quality fit ($R^2>0.94$), with the quality of the fit remaining stable as $p$ grows.
On the other hand, if smaller $N$ are included, the quality of fit begins to decay with $p$.  Therefore we include only $N\geq 28$, obtaining the fit presented in Fig.~\ref{fig:Fig2_pt_1_QAOA}B. We observe that TTS of QAOA grows as $\QAOATTS^N$ with problem size at constant QAOA depth $p=12$. We present evidence that the scaling exponent for QAOA at $p=12$ is not sensitive to the choice of $N_{\min}$ in Supplementary Information. %

As a quantum optimization heuristic with constant depth, on a fault-tolerant quantum computer the 
QAOA performance can be improved by using amplitude amplification~\cite{2208.06909,Sanders2020} or, more specifically, quantum minimum-finding~\cite{van_Apeldoorn_2020} (see Methods).
The resulting scaling of TTS of QAOA augmented with quantum minimum-finding (``QAOA$+$QMF'') is $\QAOAAATTS^N$.

We observe that, beyond a certain value ($p\approx 12$), increasing QAOA depth does not lead to better scaling of TTS.
This behavior is demonstrated in Fig.~\ref{fig:Fig2_pt_1_QAOA}C.  Consequently, running QAOA with $p$ higher than $12$ does not give any scaling advantage over amplitude amplification. This behavior is illustrated in Fig.~\ref{fig:Fig2_pt_2_QAOA}A, which shows the increase in the success probability $\popt$ from applying a given step of QAOA and amplitude amplification. 
For amplitude amplification, at step $p$ we have $\popt = \left(\sin((2p+1)\arcsin{\sqrt{p_0}})\right)^2$, where $p_0 = \frac{8}{2^N}$ is the initial (random guess) success probability~\cite{Boyer1998tight}.
Note that the $8$ in the numerator is a consequence of a dihedral group symmetry, namely, $D_{4}$. While asymptotically equivalent, amplitude amplification performs better than a realistic generalized minimum-finding algorithm~\cite{van_Apeldoorn_2020},
as the formula used here considers the scenario where we know which states to amplify (i.e., the optimal merit factor is known). We observe that for small $p$, a step (layer) of QAOA gives orders of magnitude larger increase in success probability than does a step of amplitude amplification, implying an even larger improvement over direct application of quantum minimum-finding. %
We provide additional details on comparison between QAOA and amplitude amplification in the Supplementary Information. %

We observe that the QAOA dynamics with parameters optimized for expected solution quality $\langle C\rangle_{\text{MF}}$ and success probability $\popt$ are different. We plot the optimized parameters in Fig.~\ref{fig:Fig2_pt_2_QAOA}B. We note that the parameters optimized with respect to one metric give performance that is far from optimal with respect to the other metric. This can be seen in Fig.~\ref{fig:Fig2_pt_2_QAOA}C, which plots the energy distribution (with respect to the cost Hamiltonian) of the states appearing in the QAOA wavefunction weighted by probability. With the parameters optimized for $\langle C\rangle_{\text{MF}}$, the QAOA output distribution is concentrated around its mean, and the overlap with the ground state or $\popt$ is very small. On the other hand, when the parameters are optimized with respect to $\popt$, the wavefunction is not concentrated and has large probability weight on the target ground state (i.e., high $\popt$). This comes at the cost of substantial overlap with high-energy states, which leads to poor expected solution quality. 
In the Supplementary Information, we discuss the behavior of QAOA with parameters optimized with respect to different objectives. %

\begin{figure*}
    \centering
    \includegraphics[width=1\textwidth]{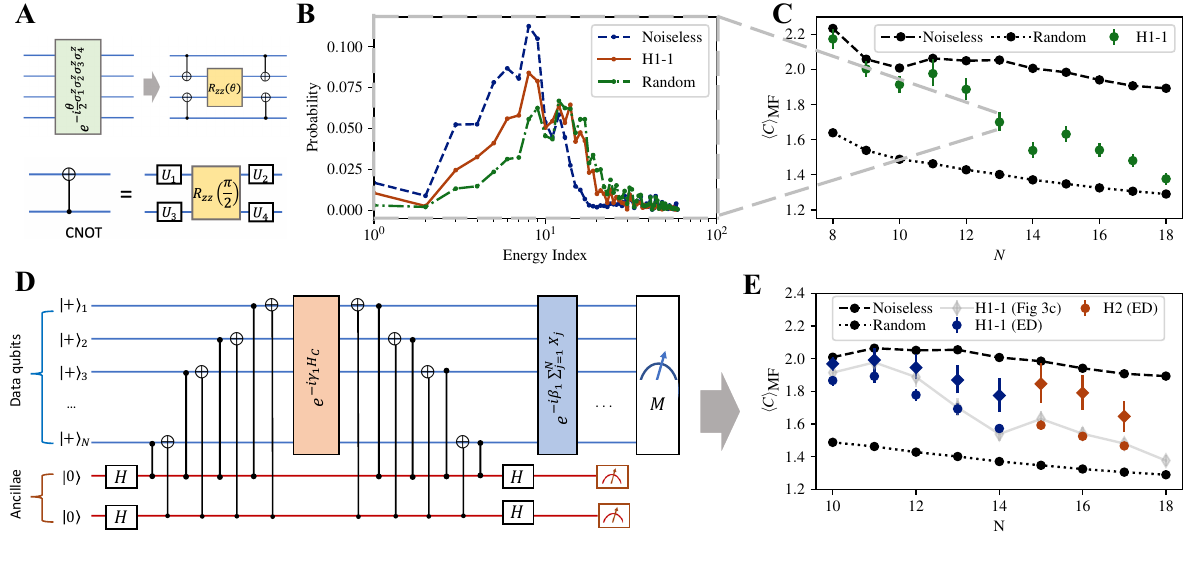}
    \caption{
    \textbf{Experimental results on trapped-ion system.} \textbf{A,} Decomposition of four-body interaction terms into a two-body $R_{\zgate\zgate}$  gate and four two-body $\cnotgate$ gates, which can be realized via native $R_{\zgate\zgate}$ gates. %
    \textbf{B,} Energy density plot from experimental measured bitstrings for $N$=13. Energy index is arranged in energy ascending order. As a comparison, the distributions for noiseless $p=1$ QAOA simulation and random guess (assuming uniform distribution of all possible bitstrings) are shown.  \textbf{C,} Experimental results up to 18 qubits on a trapped-ion quantum device (H1-1) with QAOA layer $p=1$ with optimized QAOA parameters. %
    The error bars are calculated with 99\% confidence intervals hereafter.
    \textbf{D,} Illustration of parity check circuit. The $\zgate$ and $\xgate$ parities of states are mapped to ancillary qubits after implementation of full (or part of) phase operators via $\czgate$ and $\cnotgate$ gates, respectively, followed by mid-circuit measurement on the ancillary qubits to extract the parity syndrome result. \textbf{E,} Experimental results for circuit with parity check. Three mid-circuit $\zgate$-parity and $\xgate$-parity checks were performed using six ancillary qubits. The ancillae can also be reused after appropriate reset during the circuit. The red data points were run on the Quantinuum H2 hardware while the blue data were from the H1-1 device. Data run on the H1-1 device without any ancillary qubits are shown in grey. Circles (diamonds) are the data without (with) post-selection. The abbreviation ED refers to the error detection via the parity checks.
    Number of mid-circuit parity checks is fixed to be two for $N=10, 11$ and three for all other $N$. 
    Improvement in expected merit factor after post-selection according to parity syndrome measurement is observed.} %
    \label{fig:Fig3_experiment}
\end{figure*}

\subsection*{Experiments on trapped-ion system}\label{sec:experiment}

We now present the experimental results demonstrating the algorithmic and hardware progress toward the practical implementation of QAOA. 
Implementation of the phase operator %
is especially challenging for currently available quantum processors. 
It requires a large number of geometrically nonlocal two-qubit gates,
 demanding high gate fidelity.

Recent progress in trapped-ion platforms based on the QCCD architecture~\cite{Moses2023RaceTrack,Pino2021} has led to a rapid increase in the number of qubits while maintaining high fidelity, enabling large-scale QAOA demonstrations~\cite{2303.02064,He2023AlignmentQAOA}, These systems implement two-qubit gates between arbitrary pairs of qubits by transporting ions into physically separate gate zones, resulting in high-fidelity two-qubit gates with  low crosstalk. 
We leverage this progress to execute QAOA circuits for the LABS problem on Quantinuum H-series trapped-ion systems. %

To implement the QAOA circuit shown in Fig.~\ref{fig:Fig1_scheme_summary}D, we have to implement the phase operator. The four-body terms in the phase operator are decomposed into $\cnotgate$ gates and the native $R_{\zgate\zgate}(\theta) = e^{-i \frac{\theta}{2} \zgate\zgate}$ rotation as shown in Fig.~\ref{fig:Fig3_experiment}A.
To reduce the cost of implementing both the two-qubit and four-qubit interaction terms, we optimize the circuit by greedily canceling $\cnotgate$ gates (for algorithm details and gate count reduction see the Supplementary Information).
The resulting circuit containing $\cnotgate$s and $R_{\zgate\zgate}$s is then transpiled
into the two-qubit $R_{\zgate\zgate}$ gates and single-qubit 
gates that can be natively implemented by the trapped-ion system. We remark that 
the number of two-qubit gates 
is $\approx 10^3$ at $N = 18$, putting our experiments among the largest quantum optimization demonstrations on quantum hardware to date.\cite{2306.03238,Pelofske2023,Niroula2022,He2023AlignmentQAOA,2303.02064}

In this work we execute QAOA circuits 
with $p=1$ 
using parameters $\bm\beta$, $\bm\gamma$ optimized in noiseless simulation,
followed by a projective measurement in the computational basis. In Fig.~\ref{fig:Fig3_experiment}B, we show the energy probability distribution of measured bitstrings for $N=13$. We observe a broad distribution due to the limited number of layers and experimental imperfections. Nevertheless, even at high $N$, where two-qubit gate count is high and the gate errors can be substantial, we observe a clear signal that indicates that QAOA is outperforming random guess. This is shown in Fig.~\ref{fig:Fig3_experiment}C, which presents the experimentally obtained expected merit factors for various problem size up to $N=18$. We note that the merit factor drops quickly for larger $N$ and is approaching random guess because of experimental imperfections. We also note that at this scale LABS is easy for classical heuristics, which obtain optimal merit factors in $<1$ second. Implementing QAOA for LABS instances that are hard for classical solvers would likely require error correction as the current implementation leads to an estimated two-qubit gate count of $\approx 7.5\times 10^5$ already at $N=67$ and $p=12$.

To improve the performance in the presence of noise, we implement an algorithm-specific error detection scheme. Since only the phase operator requires two-qubit gates, we focus on detecting errors that occur in the corresponding part of the circuit. Our scheme is based on the Pauli sandwiching error-detecting procedure of Ref.~\onlinecite{Gonzales2023}, which uses pairs of parity checks to detect some but not necessarily all errors that occur in a given part of the circuit. %
Following Refs.~\onlinecite{Shaydulin2021EM,Kakkar2022}, we use the symmetries of the optimization problem to construct the parity checks. Specifically, we note that the LABS Hamiltonian
preserves both $\zgate$ and $\xgate$ parities, that is, $[H_C, \otimes_i^N \zgate_i] = [H_C, \otimes_i^N \xgate_i] = 0$. We compute the parities onto ancillary qubits and 
perform mid-circuit measurement to determine whether 
an odd number of 
$\zgate$- or $\xgate$-flip errors occur during the circuit execution. The circuit with one check is shown in Fig.~\ref{fig:Fig3_experiment}D. In the hardware experiments shown in Fig.~\ref{fig:Fig3_experiment}E, we use up to three parity checks and
observe consistent improvements in QAOA performance after postselecting on their outcomes. After postselection, the difference of merit factor between experimental results and noiseless simulation is reduced by $54\%$ on average and up to $65\%$ for specific $N$. In the Supplementary Information we present additional details on the error-detecting scheme performance, including how performance improves with the number of parity checks and the reduction in the algorithm runtime. We note that while error detection does not directly give samples with better merit factors, the potential improvement in runtime can be translated into performance gains at the algorithm level, for example by being able to take more samples within a given time budget. In our experiments, in all but two cases the optimal bitstring could be found within the post-selected sample, and in all cases within the total sample.

\section*{Discussion}\label{sec:conclusion}

Our main finding is that quantum minimum-finding enhanced with QAOA scales better than the best known classical heuristics for the LABS problem. This provides evidence
for the potential of QAOA to act as a building block that provides algorithmic speedups on an idealized fault-tolerant quantum computer. 
We envision QAOA being used in a variety of algorithmic settings, similarly to how amplitude amplification acts as a subroutine in quantum algorithms for backtracking, branch-and-bound and so on.

\rev{While our numerical evidence is only obtained by fitting instances with $N\leq 40$, there are three observations that make us optimistic the observed scaling will hold for larger $N$. First, evaluating time to solution is only possible when the solution is known, i.e. only up to 66 variables. This limitation applies equally to the classical solvers, which also report their scaling on $N\leq 66$. Therefore, for the purposes of honest comparison between classical and quantum solvers the relevant problem size is up to $66$ variables, out of which our data covers up to $40$ variables. Second, we note that for classical solvers considered, there is no change in scaling between $N\leq 40$ and $40\leq N\leq 66$ (see Fig.~\ref{fig:Fig1_scheme_summary}B). This supports our claim that QAOA scaling will also remain the same up to at least $N=66$. Finally, we note that there is a rich literature showing that QAOA performance at small finite $N$ matches the rigorously derived infinite-size-limit behavior. This includes results for MaxCut~\cite{2110.10685}, Sherrington-Kirkpatrick
model~\cite{lipics.tqc.2022.7}, $k$-spin models~\cite{Basso2022}, and random $k$-SAT~\cite{2208.06909}. The results for $k$-SAT specifically focus on the scaling of TTS~\cite{2208.06909}, matching our setting.}

We take the first step toward the execution of QAOA for the LABS problem by implementing an algorithm-specific error-detection scheme on a trapped-ion quantum processor. However, further improvements in quantum error correction and hardware are necessary to implement the quantum minimum-finding augmented with QAOA. In particular, the overheads of fault-tolerance~\cite{Babbush2021} must be substantially reduced to realize the quantum speedup.

\section*{Materials and Methods}\label{sec:methods}

\subsection*{Quantum minimum-finding enhanced with QAOA}

In this work, we present the scaling results for QAOA combined with amplitude amplification (AA), or, more specifically, with quantum minimum-finding (``QAOA$+$QMF'' in Table~\ref{tab:summary}). This reduces the scaling exponent by half as compared to directly sampling QAOA output. 
We now discuss in detail how QAOA is combined with the generalized quantum minimum-finding algorithm of Ref.~\onlinecite{van_Apeldoorn_2020} to obtain the stated scaling.

We begin by noting that standard AA is not sufficient. This is because the LABS problem is framed as optimization and not search, i.e. there is no oracle for marking a global minimum. The trick for handling optimization is to perform a standard reduction from optimization to feasibility. The reduction is performed by introducing a threshold on the cost as a constraint and performing a binary search using AA as a subroutine. The oracle used by AA  marks the elements below the current threshold. This reduction was first introduced by D\"{u}rr and H\o{}yer (DH)~\cite{quant-ph/9607014}. However, the quantum minimum-finding algorithm of D\"{u}rr and H\o{}yer utilizes standard Grover search, i.e. it requires the initial state to be the uniform superposition. A modification to it is required to leverage the improved success probability afforded by QAOA. 

Ref.~\onlinecite{van_Apeldoorn_2020} provided a simple extension of DH that allows arbitrary initial states, with the overall cost scaling inversely with overlap between the initial state and state encoding the optimal solution. We leverage this extension in our quantum algorithm. We use constant-depth QAOA to prepare the initial state for the quantum minimum-finding algorithm. As QAOA state has overlap with the optimal state that is much larger than that of uniform superposition and scales more favorably, we obtain better performance than the direct minimum-finding of D\"{u}rr and H\o{}yer. Specifically, we provide numerical evidence that our algorithm obtains a super-Grover speedup over exhaustive search for the LABS problem, and scales better than the best known classical heuristics. %
We present our modification to include QAOA for outputting an optimal solution $x^*$ to the LABS problem in Algorithm \ref{alg:qaoa_aa} below. It is based on the generalized minimum-finding procedure outlined in Lemma 48 of Ref.~\onlinecite{van_Apeldoorn_2020}. To keep the current work self-contained, we include the analysis of the algorithm below. We will use the following standard quantum subroutine based on Grover search that searches for an element with unknown probability in a quantum state.

\begin{lemma}[Exponential Quantum Search, {Ref.~\onlinecite{Brassard_2002}}]
\label{lem:eqsearch}
    Let $|\psi\rangle= U|0\rangle^{\otimes N}$ be a quantum state in a $2^N$-dimensional Hilbert space with computational basis elements indexed by $N$-bit bitstrings, and $m \colon \{0,1\}^{N} \to \{0,1\}$ be a marking function such that $\sum_{\{x | m(x) = 1\}} |\langle \psi | x \rangle|^2 \ge p$. There exists a quantum algorithm $\mathbf{EQSearch}(U,m,\delta)$ that outputs an element $x^*$ such that $m(x^*) = 1$ with probability at least $\delta$ using $O\left(\frac{1}{\sqrt{p}}\log\left(\frac{1}{\delta}\right)\right)$ applications of $U$ and $m$.
\end{lemma}

\begin{algorithm}[H]
\caption{QAOA Enhanced with Quantum Minimum-Finding}
\label{alg:qaoa_aa}
\begin{algorithmic}
\Require Unitary $U_{\text{QAOA}}$ acting on $\mathbb{C}^{2^N}$ such that $\lvert\langle x^{*}|U_{\text{QAOA}}|0\rangle^{\otimes N}\rvert \geq 1/\sqrt{p_\text{opt}}$ for unknown $p_\text{opt}$, $V_{\text{LABS}}$ for computing $\mathcal{E}_{\text{sidelobe}}$
into a register, and $\delta \in (0, 1)$, positive number $M \le 2^N$, $C$ is the constant corresponding to the $O(\cdot)$ in Lemma~\ref{lem:eqsearch}
\Ensure If $M$ is greater than $1/\sqrt{p_{\text{opt}}}$, output $x^{*}$ with $\geq 1 - \delta$ probability using $O(\log(1/\delta)M)$ calls to $U_{\text{QAOA}}$ and $V_{\text{LABS}}$ (and their inverses).
\State $x_{\text{res}}$ is set to an empty list.
\For{$i \gets 1$ \textbf{to} $\lceil \log(1/\delta) \rceil$}
\State $t \gets 0$; $s_0 \gets \infty$
\While {\text{number of calls to $U_{\text{QAOA}}$ \& $V_{\text{LABS}}$ is $<$ $3CMN$}}
\State $t \gets t + 1$
\State Define $m_t \colon \{0,1\}^{N} \to \{0,1\}$ such that $m_t(x) = 1$ if and only if $\mathcal{E}_{\text{sidelobe}}(x) < s_{t-1}$. Note that $m_t$ can be coherently evaluated using one query to $V_{\text{LABS}}$.
\State Set $s_t = \mathbf{EQSearch}(U_\text{QAOA},m_t,1/(6\cdot2^{N}))$.
\EndWhile
\State Append $s_t$ to $x_{\text{res}}$.
\EndFor
\State \text{Output minimum of $x_{\text{res}}$.}
\end{algorithmic}
\end{algorithm}

\begin{figure*}
    \centering
    \includegraphics[width=\textwidth]{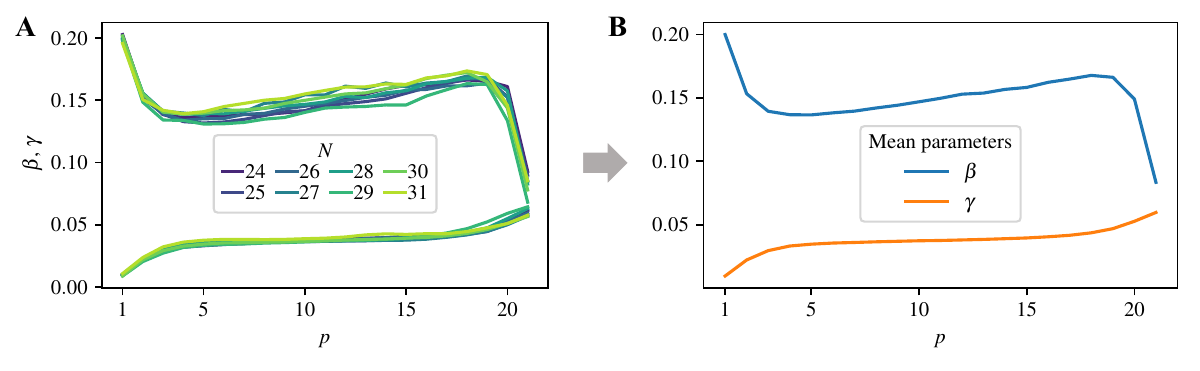}
    \caption{\textbf{Visualization of how the fixed parameters are obtained.} \textbf{A,} Optimized QAOA parameters $\bm\beta$ (top lines) and $\bm\gamma$ (bottom lines) for $p=21$. $\bm\gamma$ is multiplied by $N/24$ (constant factor of $\frac{1}{24}$ added for figure readability in both subfigures). 
    \textbf{B,} Fixed parameters obtained by taking the arithmetic mean over the optimized parameters. 
    \label{fig:mean_parameters}
    }
\end{figure*}

\begin{theorem}
Suppose a constant-depth QAOA circuit $U_{\text{QAOA}}$ prepares a state $|\psi\rangle = U_{\text{QAOA}}|0\rangle^{\otimes N}$ with $N \ge 3,$ such that we have $\lvert\langle x^{*}|\psi\rangle\rvert \geq 1/\sqrt{p_{\text{opt}}}$, where $|x^*\rangle$ encodes an optimal  solution to the $N$-bit LABS problem in a computational basis state, and we assume that $p_{\text{opt}} \ge 1/N$. Then, running Algorithm \ref{alg:qaoa_aa} with parameters $M \ge 1/\sqrt{p_\text{opt}}$ and failure probability $\delta$, runs with a gate complexity of 
$O(\text{poly}(N)\log(1/\delta)M)$ and finds $x^{*}$ with probability at least $1 - \delta$.
\end{theorem}

\begin{proof}
    See Supplementary Information.
\end{proof}

\subsection*{Choice of QAOA parameters \boldmath{$\beta, \gamma$}}\label{sec:qaoa_params}

Our strategy for setting the QAOA parameters $\bm\beta$, $\bm\gamma$ used in our experiments is twofold. First, we optimize QAOA parameters for small $N$ using the  FOURIER reparameterization scheme of Ref.~\onlinecite{zhou2020quantum}. Second, we use the optimized parameters for small $N$ to compute fixed QAOA parameters that are then used for larger $N$. To apply the fixed parameters to an instance with a given size $N$, we rescale the parameters $\bm\gamma$ by $N$.
We discuss the parameter optimization scheme and the parameter rescaling in the Supplementary Information.
We note that the results presented above can be improved if better parameter setting strategies are used.

The procedure for obtaining the set of fixed QAOA parameters is visualized in Fig.~\ref{fig:mean_parameters}. Specifically, we optimize QAOA parameters for a set of small instances with sizes $\{N_j\}_{j=1}^M$ attainable in simulation and set the fixed parameters to be the mean over the optimized parameters: 
\begin{align}
    \bm\beta^{\text{Fixed}} & = \frac{1}{M}\sum_{j = 1}^M \bm\beta^{*}_{N_j}, \\ 
    \bm\gamma^{\text{Fixed}} & = \frac{1}{M}\sum_{j = 1}^M N_j\bm\gamma^{*}_{N_j},
\end{align}
where $\bm\beta^{*}_{N_j}$, $\bm\gamma^{*}_{N_j}$ are the QAOA parameters optimized for the LABS instance of size $N_j$ and $M$ is the number of optimized instances. Then the parameters used in QAOA for size $N$ are given by $\bm\beta^{\text{Fixed}}, \frac{ \bm\gamma^{\text{Fixed}}}{N}$. We use $24\leq N_j\leq 31$ ($M=8$). 

\subsection*{Error detection by symmetry verification}

The error detection scheme relies on the symmetry of phase operator defined by Eq.~\ref{eq:labs_en}. As it commutes with both $\otimes_i^N \zgate_i$ and $\otimes_i^N \xgate_i$ operators, one can measure the value of these operators and perform postselection on the measurement outcomes. That is, the state after the phase operator should have the same $\zgate$ and $\xgate$ parity as before it. In the presence of an odd number of bit flip or phase flip errors that occur during the implementation of phase operators, the resulting state will not be in the +1 eigenspace of the two syndrome operators. 

Experimentally, we divide the whole phase operator into $m$ splits such that each split has approximately the same number of two-qubit gates, and we perform syndrome checks at the end of each split to detect errors. The syndrome operators are mapped to ancillary qubits via sequential controlled-$\xgate$ or controlled-$\zgate$ gates and Hadamard gates applied before and after the partial phase operator.
Since the number of two-qubit gates for the phase operators is higher than the number of gates used for the mapping by 2--3 orders of magnitude, additional errors introduced by ancillae  are negligible.
Furthermore, the crosstalk error probability during mid-circuit measurements is on the order of $10^{-5}$, considerably lower than the typical two-qubit-gate infidelity of $2\times 10^{-3}$ for the trapped-ion systems we used~\cite{Moses2023RaceTrack}. As a result, our error detection scheme leads to large improvements in QAOA performance on hardware at the cost of the number of repetitions growing exponentially with the number of checks~\cite{Gonzales2023}.
We note that the performance of the error detection scheme can be further improved by implementing parity checks using fault-tolerant constructions~\cite{Self_arxiv2022protecting}.

\subsection*{Scaling of classical solvers}

All scaling coefficients are obtained by fitting a least-squares linear regression on the logarithm of TTS. The confidence intervals on the scaling coefficients are obtained by using the Student's $t$-distribution and are reported with $95\%$ confidence.

We use commercial state-of-the-art branch-and-bound solvers in numerical experiments. Figure~\ref{fig:Fig1_scheme_summary}B and Table~\ref{tab:summary} show results obtained using Gurobi~\cite{Gurobi}, although we obtain similar results for CPLEX~\cite{cplex2021v20.1} (see the Supplementary Information). The use of commercial branch-and-bound solvers is motivated by the observation that their scaling closely matches that reported in Ref.~\onlinecite{Packebusch2016}. Specifically, we observe that for both solvers the time to produce a certificate of optimality (TTO) scales with an exponent within  a $95\%$ confidence interval of the $1.73$ exponent reported in Ref.~\onlinecite{Packebusch2016}. We note that unlike the solver presented in Ref.~\onlinecite{Packebusch2016}, commercial solvers are not parallelizable and can take advantage of only one CPU with at most tens of cores. Since QAOA is a heuristic and does not guarantee optimality, we additionally run branch-and-bound solvers until a solution with an exactly optimal merit factor is found,
at which point the execution is stopped. The resulting TTS scales more favorably: for Gurobi, the scaling is $1.615^N$, with a $95\%$ confidence interval of $(1.571,1.659)$. All the numbers reported correspond to the mean CPU time, with the mean taken over 100 random seeds for $N\leq 32$ and 10 random seeds for $N>32$. We present additional details of classical solver benchmarking in the Supplementary Information.

Branch-and-bound algorithms are the best-known exact solvers for the LABS problem. In the regime where proving optimality is out of reach and the goal is simply to efficiently obtain sequences with high merit factors, heuristic algorithms are preferable. The best runtimes and runtime scaling reported in the literature~\cite{Gallardo2009} are from an algorithm known as Memetic Tabu. Memetic Tabu is a \emph{memetic algorithm}, that is, an evolutionary algorithm augmented by local search. Specifically, an evolutionary algorithm is used to find initializations for tabu search, a metaheuristic that augments local neighborhood search with a data structure (known as the tabu list) that filters possible 
local moves if the potential solutions have been recently visited or diversification rules are violated~\cite{Glover:TabuSearch}. In terms of the runtime required to find optimal solutions in the regime where exact solutions have been found using branch-and-bound methods~\cite{Packebusch2016}, Memetic Tabu has been observed to outperform both non-evolutionary methods as well as memetic algorithms that use simpler neighborhood search schemes such as steepest descent. To verify the scaling of tabu search on the regime of interest for comparison with QAOA, we use the implementation of Memetic Tabu in Ref.~\onlinecite{OPUS2-git_labs-Boskovic}. For each length, we average the runtime over 50 random seeds, obtaining the scaling of the time-to-solution of $1.35^N$ with a $95\%$ confidence interval of $(1.33,1.38)$. This scaling closely matches the one reported in Ref.~\onlinecite{Bokovi2017}. We also note that solvers based on self-avoiding random walks~\cite{OPUS2-git_labs-Boskovic} have been shown to be competitive with or outperform Memetic Tabu when the task is to find skew-symmetric sequences with the lowest autocorrelation. These solvers are specialized to search for skew-symmetric sequences and do not naturally extend to the unrestricted LABS problem.

\subsection*{High-performance simulation of QAOA}

Our numerical results are enabled by a custom scalable high-performance algorithm-specific QAOA simulator. We briefly describe the simulator here; for additional details and benchmarks comparing the developed simulator with the state-of-the-art methods for simulating QAOA the reader is referred to Ref.~\onlinecite{FUR}.

In this work, the main goal of the numerical simulation of QAOA is to evaluate the expectation of the cost Hamiltonian $\langle C\rangle_{\text{MF}}$ and the success probability $\popt$. Since $\popt$ is exponentially small, it has to be evaluated with high precision. While many techniques can be leveraged for exact simulation, we opt to directly simulate the full quantum state as it is propagated through the QAOA circuit. We note in particular that tensor network techniques do not provide a benefit in this case since the circuit we simulate is deep and fully connected (see Ref.~\onlinecite{FUR} for detailed comparison). 

First, we leverage the observation that the cost Hamiltonian and hence the phase-separation operator are diagonal. This allows us to precompute the cost function evaluated at every binary input and multiply the exponentiated costs elementwise with the statevector to simulate the application of the phase-separation operator. This operation can be easily parallelized since it is an elementwise operation local to each element in the statevector. The same precomputed vector of cost function values is used to compute $\langle C\rangle_{\text{MF}}$ by taking the inner product with the final QAOA state. The cost of precomputation is amortized over the large number of objective evaluations performed during parameter optimization and is thereby negligible.

Second, we note that the mixing operator consists of an application of a uniform $\xgate$ rotation applied on each qubit. Therefore, each  rotation operation can be computed by multiplying a fixed $2 \times 2$ unitary matrix with a $2 \times 2^{n-1}$ matrix constructed from reshaping the statevector. This step is parallelized by grouping the pairs of indices on which the $2 \times 2$ unitary is applied.

We perform the simulations on the Polaris supercomputer located in Argonne Leadership Computing Facility. We distribute the simulation to 256 Polaris nodes with four NVIDIA A100 GPUs on each node and one AMD EPYC CPU.
The CPU is used to manage the communication and the assembly of final results. %
Each GPU hosts a chunk of the full statevector and a chunk of the integer cost operator vector. Application of the cost operator does not require any communication since it is local to each element.
The grouping in the mixing operator depends on index $i$ of the operator $\xgate_i$ analogous to the grouping in the fast Walsh--Hadamard transform~\cite{Fino1976}. For $i\leq n - \log_2(1024) = 29$ the pairing is local within each GPU.
For $i>29$ we use CUDA-enabled MPI to distribute full chunks between nodes, which requires space to be reserved for two statevector chunks on each GPU.

\section*{Author contributions}

R.~Shaydulin devised the project. J.~Larson, N.~Kumar, and R.~Shaydulin implemented QAOA parameter optimization and the parameter setting schemes. R.~Shaydulin and Y.~Sun implemented the single-node version of the QAOA simulator. D.~Lykov implemented the distributed version of the QAOA simulator and executed the large-scale simulations on Polaris. R.~Shaydulin analyzed the simulation results. D.~Herman and S.~Hu developed the circuit optimization pipeline. C.~Li implemented and analyzed the error detection scheme. M.~DeCross and D.~Herman executed the experiments on trapped-ion hardware, and C.~Li analyzed the results. S.~Chakrabarti, D.~Lykov and P.~Minssen benchmarked classical solvers. S.~Chakrabarti, D.~Herman and R.~Shaydulin analyzed the generalized quantum minimum-finding enhanced with QAOA. J. Dreiling, J.P. Gaebler, T.M. Gatterman, J.A. Gerber, K. Gilmore, D. Gresh, N. Hewitt, C.V. Horst, J. Johansen, M. Matheny, T. Mengle, M. Mills, S.A. Moses, B. Neyenhuis, and P. Siegfried built, optimized, and operated the trapped-ion hardware. M.~Pistoia led the overall project. All authors contributed to technical discussions and the writing of the manuscript.

\begin{acknowledgments}
The authors thank their colleagues at the Global Technology Applied Research center of JPMorgan Chase for support and helpful discussions. Special thanks are also due to Tony Uttley and Jenni Strabley from Quantinuum for their continued support throughout the project. 

This paper was prepared for informational purposes with contributions from the Global Technology Applied Research center of JPMorgan Chase \& Co., Argonne National Laboratory and Quantinuum LLC. This paper is not a product of the Research Department of JPMorgan Chase \& Co. or its affiliates. Neither JPMorgan Chase \& Co. nor any of its affiliates makes any explicit or implied representation or warranty and none of them accept any liability in connection with this position paper, including, without limitation, with respect to the completeness, accuracy, or reliability of the information contained herein and the potential legal, compliance, tax, or accounting effects thereof. This document is not intended as investment research or investment advice, or as a recommendation, offer, or solicitation for the purchase or sale of any security, financial instrument, financial product or service, or to be used in any way for evaluating the merits of participating in any transaction.

The submitted manuscript includes contributions from 
UChicago Argonne, LLC, Operator of Argonne National Laboratory (``Argonne'').
Argonne, a U.S.\ Department of Energy Office of Science laboratory, is operated
under Contract No.\ DE-AC02-06CH11357.  The U.S.\ Government retains for itself,
and others acting on its behalf, a paid-up nonexclusive, irrevocable worldwide
license in said article to reproduce, prepare derivative works, distribute
copies to the public, and perform publicly and display publicly, by or on
behalf of the Government.  The Department of Energy will provide public access
to these results of federally sponsored research in accordance with the DOE
Public Access Plan \url{http://energy.gov/downloads/doe-public-access-plan}.

\subsection*{Funding}
This material is based upon work supported in part by the U.S. Department of Energy, Office of Science, under contract number DE-AC02-06CH11357
and the Office of Science, Office of Advanced Scientific Computing Research, Accelerated Research for Quantum Computing program.

\end{acknowledgments}

\section*{Data availability}

The data presented in this paper can be found at \url{https://doi.org/10.5281/zenodo.8190275}.
\vspace{0.2in} %
\section*{Code availability}

The code used to produce the results in this paper can be found at \url{https://zenodo.org/doi/10.5281/zenodo.10935810}. The latest version can be found at \url{https://github.com/jpmorganchase/QOKit}.

\bibliographystyle{unsrturl}
\bibliography{ms}

\onecolumngrid
\newpage
\section*{Supplementary Materials for\\ Evidence of Scaling Advantage for the Quantum Approximate \\ Optimization Algorithm on a Classically Intractable Problem}

\setcounter{section}{0}
\setcounter{theorem}{0}
\setcounter{lemma}{0}
\setcounter{algorithm}{0}
\setcounter{equation}{0}
\setcounter{figure}{0}
\setcounter{table}{0}
\setcounter{page}{1}
\makeatletter
\renewcommand{\thesection}{S\arabic{section}}
\renewcommand{\theequation}{S\arabic{equation}}
\renewcommand{\thefigure}{S\arabic{figure}}
\renewcommand{\thetable}{S\arabic{table}}
\makeatother

\section{Background on the LABS problem}

The problem of finding sequences with low sidelobe energies attracted attention in the 1960s and 1970s due to its applications to the reduction of the peak power of radar pulses~\cite{Boehmer1967,Schroeder1970}. The merit factor $F$ was first introduced by Golay~\cite{Golay1972} and defined as the ratio of central to sidelobe energy of a binary sequence. Improved merit factors were obtained over the years by exhaustive \cite{Golay1977,Golay1990} and non-exhaustive~\cite{beenker1985binary} search methods. Explicit sequences asymptotically achieving the merit factor of $\mathcal{F}\approx 6.34$ are known~\cite{Jedwab2013}. The conjectured asymptotic limit of $\mathcal{F}\rightarrow 12.32$ as $N\rightarrow \infty$ was derived using arguments from statistical mechanics in Ref.~\onlinecite{Golay1982}. Bernasconi reframed the LABS problem as a spin model with long-range 4-spin interactions to apply simulated annealing to it~\cite{Bernasconi1987}. However, simulated annealing failed to obtain high-quality solutions, with the failure attributed to the ``golf-course type'' energy landscape~\cite{Bernasconi1987}. Bernasconi further conjectured that this property of the landscape will prevent stochastic search procedures from obtaining high-quality solutions for long sequences~\cite{Bernasconi1987}. \rev{While the presence of isolated deep and narrow (``golf-course type'') global minima has been debated~\cite{Ferreira2000}, the conjecture that stochastic search procedures will fail} has held up so far.

A commonly considered class of sequences are those exhibiting skew-symmetry, which for sequences of odd length $N=2k-1$ are defined as $s_{k+l} = (-1)^ls_{k-l},\; l\in \{1, \ldots, k-1\}$. 
Skew-symmetric sequences are known to be optimal for many odd $N$ instances. Since skew-symmetry reduces the search space from $2^N$ to $2^{\frac{N}{2}}$, only searching this subspace leads to better runtime scaling. Therefore many algorithms are restricted to only searching this subspace. The best-known heuristic for skew-symmetric LABS uses a sequence of self-avoiding walk segments when searching and has a running time that scales as $1.15^N$.\cite{Bokovi2017} In this work, we target the general LABS problem and therefore we do not consider solvers that are only capable of tackling the skew-symmetric instances.

\rev{To identify the best classical solvers for the LABS problem, we have performed an extensive literature review. Table~\ref{tab:classical_solvers_for_labs} summarizes this review. We also note recent works targeting a peak sidelobe value rather than merit factor~\cite{Dimitrov2020,Brest2021,bovskovic2021two}. The ideas from these works may be fruitfully applied to the LABS problem in the future.}

\begin{table}[H]
    \centering
    \rev{\begin{tabular}{c|c|c|c|c|c|c}
        Year & Ref. & Kind & General & Scaling & $N_{\max}$ & Algorithm \\
        \hline
        1975 & \onlinecite{Lindner1975} & Exact & Yes & N/A & 40 & Exhaustive search  \\
        1977 & \onlinecite{Golay1977} & Exact & No & N/A & 59 & Exhaustive search \\ %
        1985 & \onlinecite{beenker1985binary} & Heuristic & No & N/A & 199 & Non-exhaustive search \\
        1990 & \onlinecite{Golay1990} & Exact & No & N/A & 69 & Exhaustive search \\
        1990 & \onlinecite{Golay1990} & Heuristic & No & N/A & 117 & Non-exhaustive search \\
        1996 & \onlinecite{Mertens1996} & Exact & Yes & $1.85^N$ & 48 & Branch-and-bound \\
        2003 & \onlinecite{Brglez2003} & Heuristic & Yes & $1.463^N$ & 64 & Kernighan-Lin \\ %
        2003 & \onlinecite{Brglez2003} & Heuristic & Yes & $1.397^N$ & 47 & Evolutionary \\ %
        2006 & \onlinecite{10.1007/11889205_51} & Heuristic & Yes & N/A & 48 & Tabu search \\ 
        2009 & \onlinecite{Gallardo2009} & Heuristic & No & $1.17^N\;$ \cite{Bokovi2017} & 201 & Memetic Tabu \\
        \textbf{2009} & \textbf{\onlinecite{Gallardo2009}} & \textbf{Heuristic} & \textbf{Yes} & $\mathbf{1.34^N}\;$ \cite{Bokovi2017} & 85 & \textbf{Memetic Tabu} \\
        2012 & \onlinecite{kratica2012mixed} & Exact & No & N/A & 30 & Mixed-integer quadratic programming \\ %
        2012 & \onlinecite{kratica2012mixed} & Exact & Yes & N/A & 51 & Mixed-integer quadratic programming \\ %
        2013 & \onlinecite{Prestwich2013} & Exact & No & $1.337^N\;$ \cite{Bokovi2017} & 89 & Branch-and-bound \\
        \textbf{2016} & \textbf{\onlinecite{Packebusch2016}} & \textbf{Exact} & \textbf{Yes} & $\mathbf{1.73^N}$ & 66 & \textbf{Branch-and-bound} \\
        2017 & \onlinecite{Bokovi2017} & Heuristic & No & $1.15^N$ & 400 & Self-avoiding walks \\ 
        2018 & \onlinecite{Brest2018} & Heuristic & No & $1.18^N$ & 225 & Stochastic search
    \end{tabular}
    \caption{Select prior works tackling the LABS problem with classical solvers. Best results to date for general LABS are highlighted in bold. ``General'' column signifies whether a solver tackles general LABS or only skew-symmetric ones. The latter are only included for reference as we consider general LABS in this work. $N_{\max}$ is the largest size tackled (but not necessarly solved exactly). We note that in Ref.~\onlinecite{Brglez2003}, Kernighan-Lin solver has worse scaling but better constants, which is why it was used for larger $N$ experiments. We do not include in this table negative results, such as results showing that simulated annealing~\cite{Bernasconi1987,Ferreira2000} and plain evolutionary algorithms~\cite{Gallardo2007} fail to obtain good solutions.}}
    \label{tab:classical_solvers_for_labs}
\end{table}

\begin{figure*}[b]
    \centering
    \includegraphics[width=\textwidth]{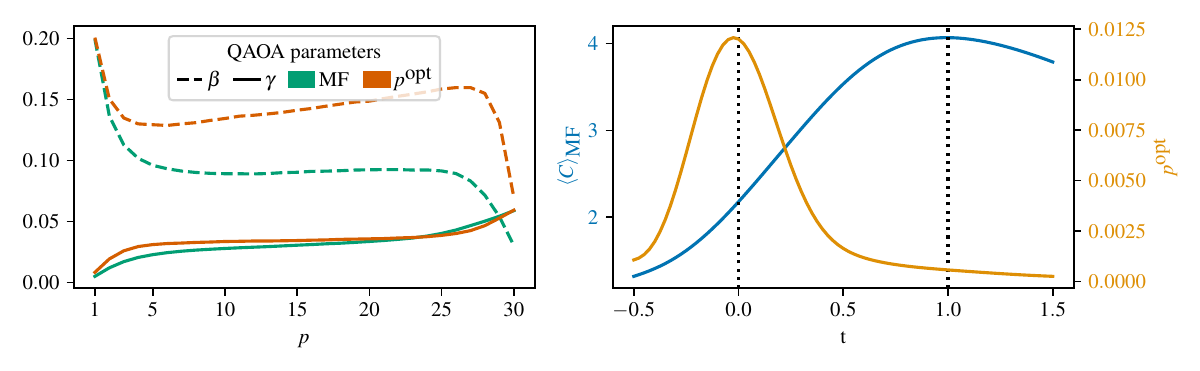}
    \caption{\textbf{QAOA parameters optimized for different objectives} \textbf{a,} Fixed QAOA parameters
    for $p=30$ chosen with respect to the QAOA energy $\langle C\rangle_{\text{MF}}$ (``MF'') and probability of sampling the optimal solution (``$\popt$'').  When the parameters are optimized with respect to $\popt$, the value of $\bm\beta$ is substantially larger throughout QAOA evolution. Subfigure reproduced from the main text. \textbf{b,} QAOA performance for $N=25$, $p=30$ with parameters linearly extrapolated between fixed parameters for $\popt$ ($t=0$) and $\langle C\rangle_{\text{MF}}$ ($t=1$). QAOA parameters optimized for $\langle C\rangle_{\text{MF}}$ give very poor values of $\popt$ and vice versa.}
    \label{fig:MF_vs_overlap_parameters}
\end{figure*}

\begin{figure*}[t]
    \centering
    \includegraphics[width=\textwidth]{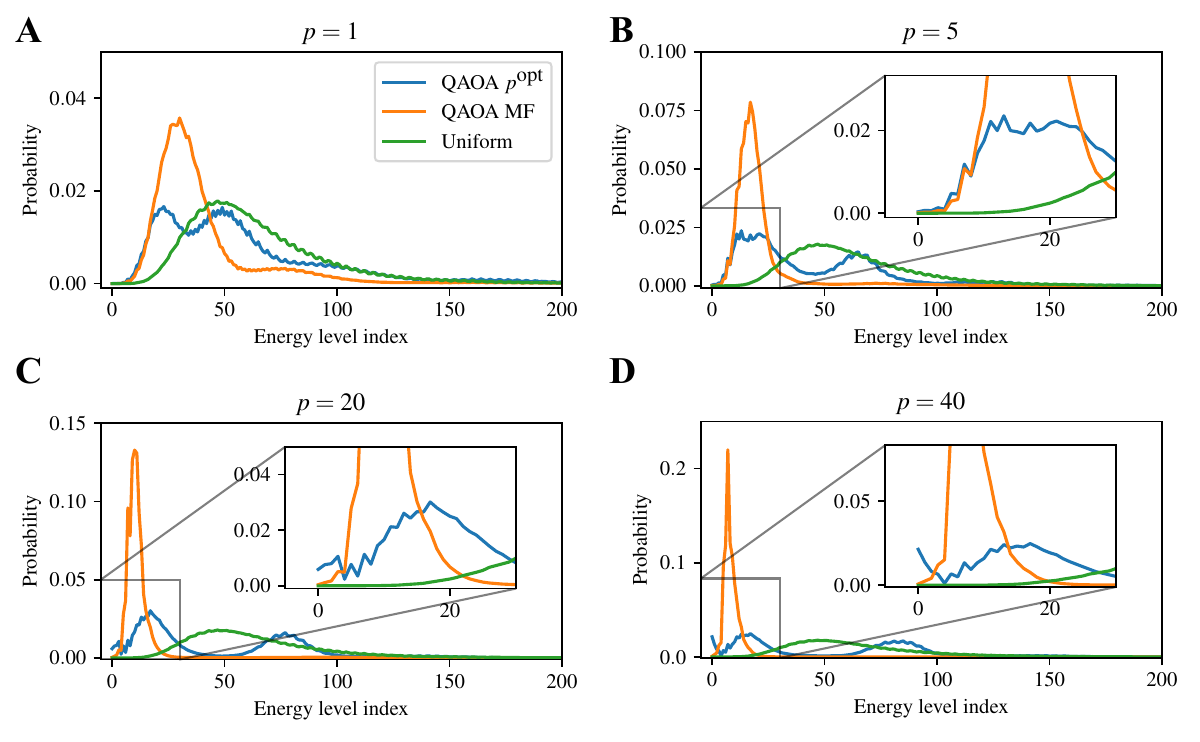}
    \caption{\textbf{QAOA dynamics with parameters optimized for different objectives} The probability of obtaining a binary string corresponding to a given energy level of the LABS problem (the zeroth energy level is the ground state or optimal solution, lower is better) for varying $p$ (\textbf{A-D}). When parameters are optimized with respect to the expected merit factor (labeled ``QAOA MF''), the QAOA output state is concentrated around the mean and fails to obtain a high overlap with the ground state. On the other hand, when parameters are optimized with respect to $\popt$ (labeled ``QAOA $\popt$''), the QAOA state has a high overlap with both ground state and higher energy states. The probability of obtaining the ground state is $27.3$ times greater for QAOA with parameters optimized with respect to $\popt$ at $p=40$ (\textbf{D}). 
    }
    \label{fig:output_distribution}
\end{figure*}

\section{QAOA as an exact and approximate optimization algorithm}

We now provide additional numerical results highlighting the differences in QAOA behavior with parameters optimized for approximate and exact solutions. In this work, we use QAOA as an exact solver, with time to solution as the target metric. However, QAOA is typically used as an approximation algorithm,~\cite{farhi2014quantum} with most theoretical results focusing on the expected solution quality obtained by QAOA. \rev{For example, this is the setting used in recent results ruling out quantum advantage in approximation with log-depth local quantum algorithms like QAOA~\cite{1905.07047,chou2021limitations,2310.01563}. These results do not apply to our chosen setting of using QAOA as an exact solver. In fact,} on the LABS problem, we observe that QAOA can provide poor approximations in polynomial time while still offering speedups as an exact exponential-time solver.  %

We begin by discussing the parameters themselves.
QAOA parameters are typically chosen with respect to QAOA energy~\cite{farhi2014quantum,zhou2020quantum,shaydulin2019multistart,crooks2018performance,streif2020training,Lee2021,sack2021quantum,amosy2022iterative} or the probability of sampling the optimal solution~\cite{2208.06909}.
Figure~\ref{fig:MF_vs_overlap_parameters}a shows the difference in QAOA parameters in the two scenarios for the specific case of the LABS problem.
First, we note that QAOA parameters that maximize $\langle C\rangle_{\text{MF}}$ are substantially different from those maximizing $\popt$. Specifically, while the values of $\bm\gamma$  are similar, the value of $\bm\beta$ for $\popt$ is much larger than that for $\langle C\rangle_{\text{MF}}$. Qualitatively, this implies that the probability amplitudes are allowed to ``mix'' more, making QAOA state not concentrated with respect to Hamming distance. This behavior is seen when examining the QAOA output distribution, which is shown for both parameter schedules in Fig.~\ref{fig:output_distribution} and discussed in detail below. %

\begin{figure*}[bt]
    \centering
    \includegraphics[width=\textwidth]{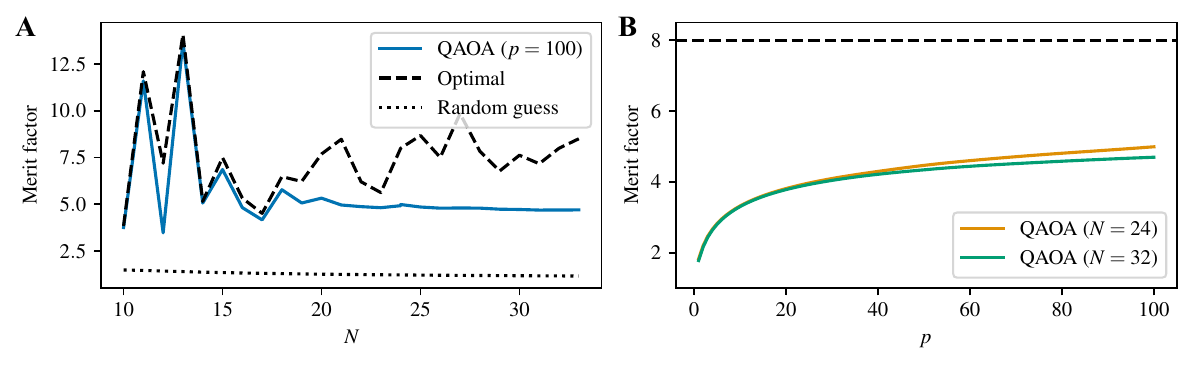}
    \caption{\textbf{QAOA performs poorly as an approximation algorithm} \textbf{A,} Performance of QAOA as an approximation algorithm. %
    Explicit constructions of skew-symmetric sequences exist that achieve $\mathcal{F}\approx 6.34$ for large $N$.\cite{Jedwab2013} Simulated annealing achieves $\mathcal{F}\approx 5$ for large $N$.\cite{Bernasconi1987} For QAOA, the expected value of merit factor $\langle C\rangle_{\text{MF}}$ is plotted. The expected merit factor of QAOA output is below the values easily attainable classically. \textbf{B,} For both $N=24$ and $N=32$, the optimal merit factor is $8$ (dashed line). %
    }
    \label{fig:QAOA_merit_factor}
\end{figure*}

Second, QAOA parameters that give good performance with respect to one metric are far from optimal with respect to the other metric. Fig.~\ref{fig:MF_vs_overlap_parameters}b shows QAOA performance with parameters linearly extrapolated between the parameters ($\bm\beta^{\langle C\rangle_{\text{MF}}}$, $\bm\gamma^{\langle C\rangle_{\text{MF}}}$) that give a high value of $\langle C\rangle_{\text{MF}}$ and ($\bm\beta^{\popt}$, $\bm\gamma^{\popt}$) that give high $\popt$ : $\bm\gamma = t\bm\gamma^{\popt}+(1-t)\bm\gamma^{\langle C\rangle_{\text{MF}}}$ and $\bm\beta = t\bm\beta^{\popt}+(1-t)\bm\beta^{\langle C\rangle_{\text{MF}}}$. We note that the parameters $\bm\beta^{\langle C\rangle_{\text{MF}}}$, $\bm\gamma^{\langle C\rangle_{\text{MF}}}$ ($t=1$) give a very low value of $\popt$ and vice versa. This suggests that substantial performance gains are possible if parameters are chosen differently with respect to the two figures of merit, rather than using one as a proxy for the other as is commonly done in QAOA research.\cite{Larkin2022,Lotshaw2021} Similar observations have been made in Refs.~\onlinecite{Li2020,Barkoutsos2020},
though the difference observed between the two figures of merit is more drastic in our case due to the hardness of the problem considered.

In the numerical experiments presented in the main test, we use time to solution as the target metric. For completeness, we include the results showing QAOA performance as an approximation algorithm.
We observe that QAOA performs poorly on the LABS problem with respect to the expected merit factor $\langle C\rangle_{\text{MF}}$. Specifically, we observe that QAOA fails to outperform even simple classical techniques at high depth. Fig.~\ref{fig:QAOA_merit_factor} shows the expected merit factor of QAOA for fixed $p=100$ as a function of $N$, as well as examples of how the expected merit factor grows with $p$ for two fixed values of $N$. 
We observe that as $N$ grows, QAOA at $p=100$ achieves an expected merit factor $\langle C\rangle_{\text{MF}}$ of less than $5$. Note that explicit analytical sequences achieving merit factor $>6$ are known~\cite{Jedwab2013}. Moreover, we see that $\langle C\rangle_{\text{MF}}$ grows increasingly slowly as $N$ and $p$ increase, suggesting that a prohibitively high value of $p$ would be required to achieve a high expected merit factor.

We can understand this behavior by examining the values of the merit factor attainable by binary strings (in physics terms, we are examining the spectrum of the LABS Hamiltonian). For the $N=25$ problem, presented in Fig.~\ref{fig:output_distribution}, only the $4$ lowest energy levels correspond to a merit factor better than $6$. This means that QAOA must concentrate all the wave function mass on a superposition of a small number of computational basis states. Due to QAOA preserving the $D_4$ symmetry of the problem~\cite{Shaydulin2021}, this is a highly entangled state, which is hard to prepare~\cite{Bravyi2020}.
As a result, QAOA requires a very large value of $p$ to obtain a high expected merit factor. If, on the other hand, we choose the probability of sampling the exact optimal solution $\popt$ (overlap between QAOA state and the ground state of the LABS Hamiltonian) as the target metric, a much lower value of $p$ is needed to obtain good success probability. Qualitatively, Fig.~\ref{fig:output_distribution} shows how this state preparation succeeds by allowing a substantial part of the QAOA state to ``leak'' to high energy levels. This can be observed by noticing how the population of energy levels $>50$ stays relatively high for QAOA with parameters optimized with respect to $\popt$, but becomes negligible if QAOA parameters are chosen with respect to $\langle C\rangle_{\text{MF}}$.

The success of QAOA in preparing states with a large overlap with the ground state of the LABS Hamiltonian (i.e. states having a high probability of measuring exactly optimal solution) motivates the choice of time to solution as the target metric for QAOA evaluation. We show the time to solution at the largest $p$ explored numerically ($p=33$) in Figure~\ref{fig:QAOA_TTS}. We remark that the QAOA succeeds at achieving high overlap at $N\leq 40$. 

Our results suggest a new way of viewing the potential of QAOA to provide algorithmic speedups and provide an important caveat to theoretical results bounding the approximations attainable by QAOA with constant depth.\cite{Basso2022} Even in the regime where the expected solution quality of QAOA is bounded, it can be still useful as a tool to obtain a high probability of measuring the exact optimal solution.

\begin{figure*}
    \centering
    \includegraphics[width=\textwidth]{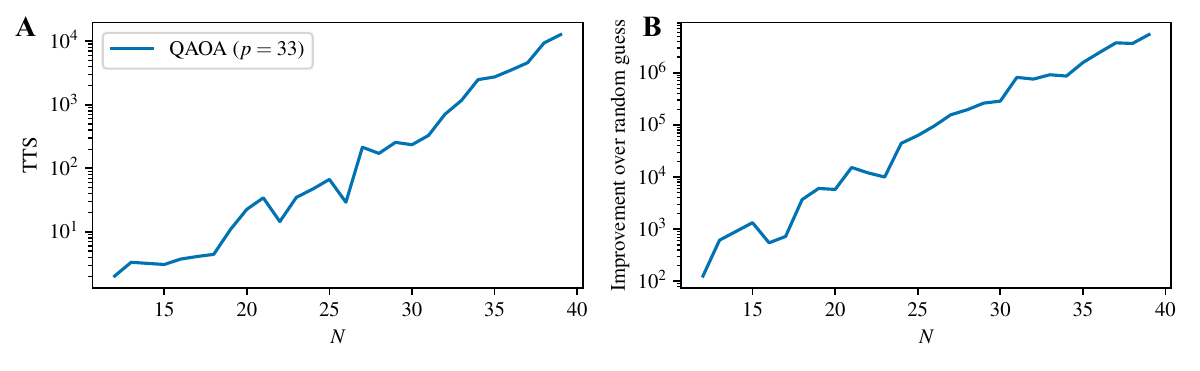}
    \caption{\textbf{QAOA achieves high overlap} \textbf{A,} Time-to-solution (TTS) of QAOA with parameters chosen with respect to $\popt$. \textbf{B,} Improvement over random guess.  For the largest case numerically considered, which is the $N=39$ problem at $p=33$, the expected number of shots required to solve it is $1.2\times 10^4$. This is a $5.4\times 10^6$ factor improvement over random guess.}
    \label{fig:QAOA_TTS}
\end{figure*}

\section{Details of numerical studies}

We now present in detail how the fixed parameters were chosen. Our procedure for doing so is as follows. First, we optimize the QAOA parameters using the FOURIER~\cite{zhou2020quantum} reparameterization. We show evidence that the optimization of the reparameterized QAOA gives the same performance as the more extensive optimization of the standard parameterization. Second, we set our fixed parameters to be the arithmetic mean of the (appropriately rescaled) optimized parameters for smaller $N$. We provide evidence that for smaller $N$ where directly optimized parameters are available, the fixed parameters lead to QAOA performance that is close to that with the optimized parameters. We note that better parameter setting schemes can only improve our performance.

\subsection{Optimized QAOA parameters for LABS change with $N$}
\begin{figure*}
    \centering
    \includegraphics[width=\textwidth]{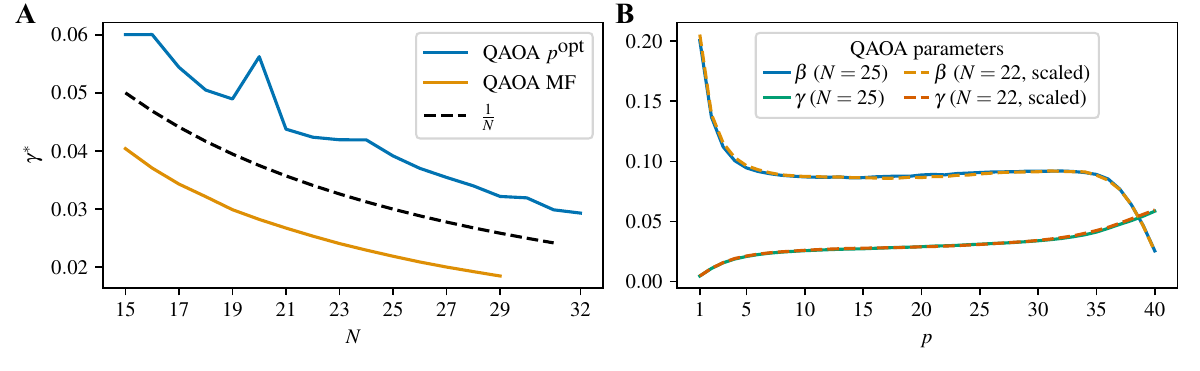}
    \caption{\textbf{Scaling of QAOA parameters with problem size} \textbf{A,} Scaling of the optimized value of QAOA parameter $\gamma$ with $N$ for $p=1$ when optimized with respect to expected merit factor $\langle C\rangle_{\text{MF}}$ (``QAOA MF'') and probability of obtaining optimal solution $\popt$ (``QAOA $\popt$''). Optimized $\gamma^*$ decreases with $N$ as $\frac{1}{N}$. \textbf{B,} QAOA parameters optimized with respect to $\langle C\rangle_{\text{MF}}$ for $N=\in\{22,25\}$. For $N=22$, the parameters $\bm\gamma$ are scaled by $22/25$. After rescaling, the parameters for $N=22$ and $N=25$ are visually indistinguishable.} %
    \label{fig:scaling_of_parameters}
\end{figure*}
First, we observe that the optimized QAOA parameters are not invariant with $N$. Specifically, we observe that the optimized value of $\bm\gamma^*$ goes down with $N$ as $\frac{1}{N}$. ~\ref{fig:scaling_of_parameters}A plots this for $p=1$. We note that $\bm\beta^*$ is roughly constant with $N$.
We observe this scaling for all $p$. As an example, ~\ref{fig:scaling_of_parameters}B plots optimized parameters for two different values of $N$. After rescaling $\bm\gamma^*$ by $N$, the two sets of parameters are visually indistinguishable. Below, we take advantage of this scaling of $\bm\gamma^*$ in two ways. First, we improve the convergence of local optimization runs by rescaling the initialization and the initial step size of the local optimizer. Second, we use it to correctly account for scale when executing QAOA with fixed parameters.

The scaling of $\bm\gamma^*$ with extremal properties of the objective function has been observed before for other problems. For example, the normalized value of the maximum cut on $D$-regular graphs grows with $D$ (to the first order) and $\bm\gamma^*$ decreases as $\frac{1}{\sqrt{D}}$.\cite{Wang2018,lipics.tqc.2022.7} %
Analogous scaling has been observed for weighted problems.\cite{Sureshbabu2023} In LABS, the energy $\mathcal{E}_{\text{sidelobe}}(\boldsymbol{s})$ grows as $N^2$, so $\bm\gamma^*$ decreases as $\frac{1}{N}$. Formally establishing this connection is a promising direction for future research. %

\subsection{QAOA parameter optimization with FOURIER scheme}

For a given figure of merit, we optimize the QAOA parameters as follows. In all cases below, we use the nlopt~\cite{NLOpt} implementation of BOBYQA~\cite{powell2009bobyqa} gradient-free local optimization algorithm. In all cases, we run BOBYQA until convergence, with convergence specified by relative tolerances on changes in parameters and in objective function value of $10^{-8}$. BOBYQA has been shown to outperform other local optimizers on the task of optimizing QAOA parameters~\cite{shaydulin2019multistart}. We expect similar results with any other local hill-climbing algorithm, albeit at a potentially different cost in terms of the number of iterations.

For $p=1$, we optimize the QAOA parameter exhaustively by running the local optimizer from $400$ initial points. We set the initial step size (\texttt{rhobeg}) to $0.01 / N$. The initial points $\beta^{\text{init}}, \gamma^{\text{init}}$ are chosen uniformly at random from $\beta^{\text{init}} \in [0.1, 0.2], \gamma^{\text{init}} \in [0, 0.85/N]$ when optimizing with respect to $\langle C\rangle_{\text{MF}}$, and $\beta^{\text{init}} \in [0.15, 0.3], \gamma^{\text{init}} \in [0.6, 1.2/N]$ when optimizing with respect to $\popt$. The regions for initializations are read off from grid search results for $p=1$.

For $p > 1$, we follow the FOURIER$[\infty,0]$ scheme of Ref.~\onlinecite{zhou2020quantum}. Specifically, we change the QAOA parameterization to the frequency domain as follows:
\begin{align}
    \gamma_i & = \sum_{k=1}^pu_k\sin{\left[\left(k-\frac{1}{2}\right)\left(i-\frac{1}{2}\right)\frac{\pi}{p}\right]}, \\
    \beta_i & = \sum_{k=1}^pv_k\cos{\left[\left(k-\frac{1}{2}\right)\left(i-\frac{1}{2}\right)\frac{\pi}{p}\right]}.
\end{align}
Then we optimize over the new parameters $\bm u,\bm v$. We take the optimized parameters $\bm u^*_{p-1},\bm v^*_{p-1}$ for $p-1$ and run one local optimization from $\bm u^*_{p} = (u^*_{p-1}, 0)$, $\bm v^*_{p} = (\bm v^*_{p-1}, 0)$. The initial step size (\texttt{rhobeg}) for local optimizer is set to $0.01 / N$. An initial step size that is small and decreasing with $N$ is central to the robust convergence of a local optimizer, due to QAOA parameters having different scales for different $N$. We do not find it necessary for obtaining high-quality parameters to perform objective function rescaling of the type explored in Refs.~\onlinecite{10.1145/3584706,Sureshbabu2023}, though we expect that it may reduce the number of iterations required by the local optimizer to converge.

\subsection{Evidence of the success of the FOURIER reparameterization heuristic}
To evaluate the success of the FOURIER parameter optimization heuristic, we compare the quality of optimized parameters it finds with the quality of the parameters obtained by running a local optimizer with $100p$ seeds from initial points sampled uniformly from $\bm\beta^{\text{init}} \in [0.1, 0.2]^p, \bm\gamma^{\text{init}} \in [0, 0.85/N]^p$. 
\begin{figure}
    \centering
    \includegraphics[width=0.5\textwidth]{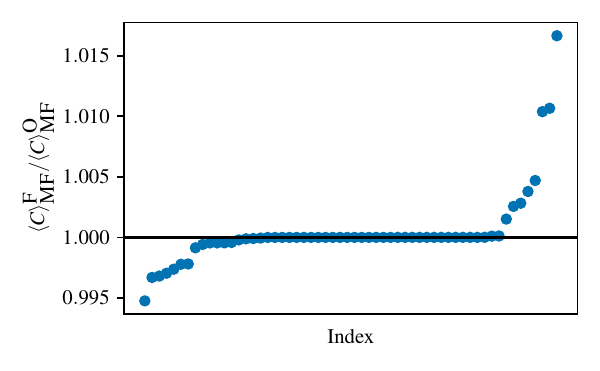}
    \caption{\textbf{FOURIER heuristic finds good parameters} The ratio between the expected merit factor of QAOA with parameters optimized by directly running local optimization from many initial points ($\langle C\rangle_{\text{MF}}^{\text{O}}$) and with parameters optimized using the FOURIER$[\infty,0]$ scheme ($\langle C\rangle_{\text{MF}}^{\text{F}}$) for $N>12$. We observe that the difference in the quality of the parameters obtained by the two optimization schemes is small.} %
    \label{fig:FOURIER_vs_directly_optimized}
\end{figure}

We find that the very expensive direct optimization performs very similarly to one local optimization run used in the FOURIER scheme, as shown in Figure~\ref{fig:FOURIER_vs_directly_optimized}. Specifically, the mean difference between the two schemes is $<0.05\%$, and in the worst case of the ones considered, FOURIER gives parameters that are only $<0.5\%$ worse. Therefore below, we simply consider parameters optimized using the FOURIER$[\infty,0]$ scheme.

\begin{figure*}
    \centering
    \includegraphics[width=\textwidth]{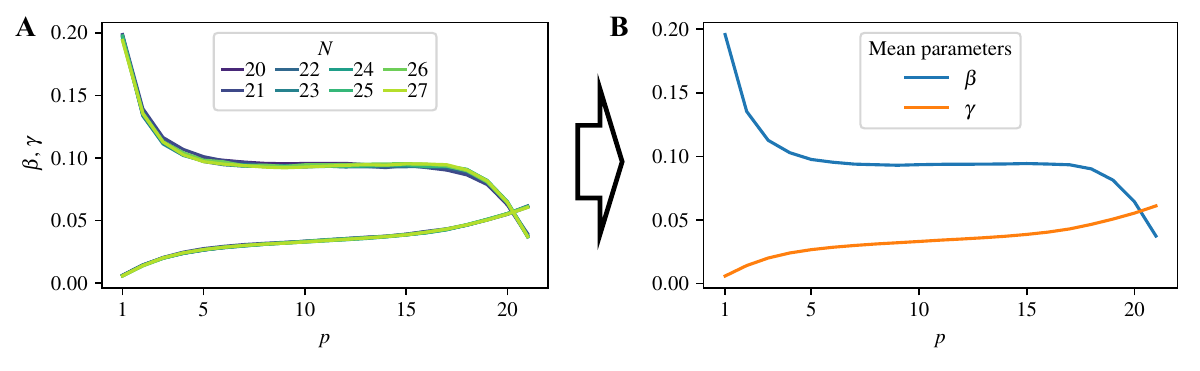}
    \caption{\textbf{Visualization of how the fixed parameters are obtained.} Visualization of how the fixed parameters with respect to $\langle C\rangle_{\text{MF}}$ are obtained. Visualization for $\popt$ is presented in the main text. \textbf{A,} Optimized QAOA parameters $\bm\beta$ (top line) and $\bm\gamma$ (bottom line) for $p=21$. $\bm\gamma$ is scaled by $N/24$, with the constant factor of $1/24$ added for figure readability. \textbf{B,} Fixed parameters obtained by taking the arithmetic mean over the optimized parameters.
    \label{fig:SI_mean_parameters}
    }
\end{figure*}

\subsection{Procedure for obtaining the fixed parameters}

The procedure we follow for obtaining the fixed parameters is described in the main text. We reiterate it here for completeness. We optimize QAOA parameters for smaller values of $N$ for which the simulation is relatively inexpensive and set the fixed parameters to be the mean over the optimized parameters:
\begin{align}
    \bm\beta^{\text{Fixed}} & = \frac{1}{M}\sum_{N \in \{N_1,\ldots, N_M\}}\bm\beta^{*}_N, \\ 
    \bm\gamma^{\text{Fixed}} & = \frac{1}{M}\sum_{N \in \{N_1,\ldots, N_M\}}N\bm\gamma^{*}_N,
\end{align}
where $\bm\beta^{*}_N$, $\bm\gamma^{*}_N$ are the QAOA parameters optimized for $N$. The fixed parameters used in QAOA applied to a LABS instance of size $N$ are then given by $\bm\beta^{\text{Fixed}}, \frac{ \bm\gamma^{\text{Fixed}}}{N}$. This process is visualized for parameters optimized with respect to $\langle C\rangle_{\text{MF}}$ in Figure~\ref{fig:SI_mean_parameters}. We use optimized parameters for $20 \leq N \leq 27$ when computing parameters for $\langle C\rangle_{\text{MF}}$ and $24 \leq N \leq 31$ for $\popt$.

\subsection{Evidence of the success of the fixed parameter scheme}

\begin{figure*}
    \centering
    \includegraphics[width=\textwidth]{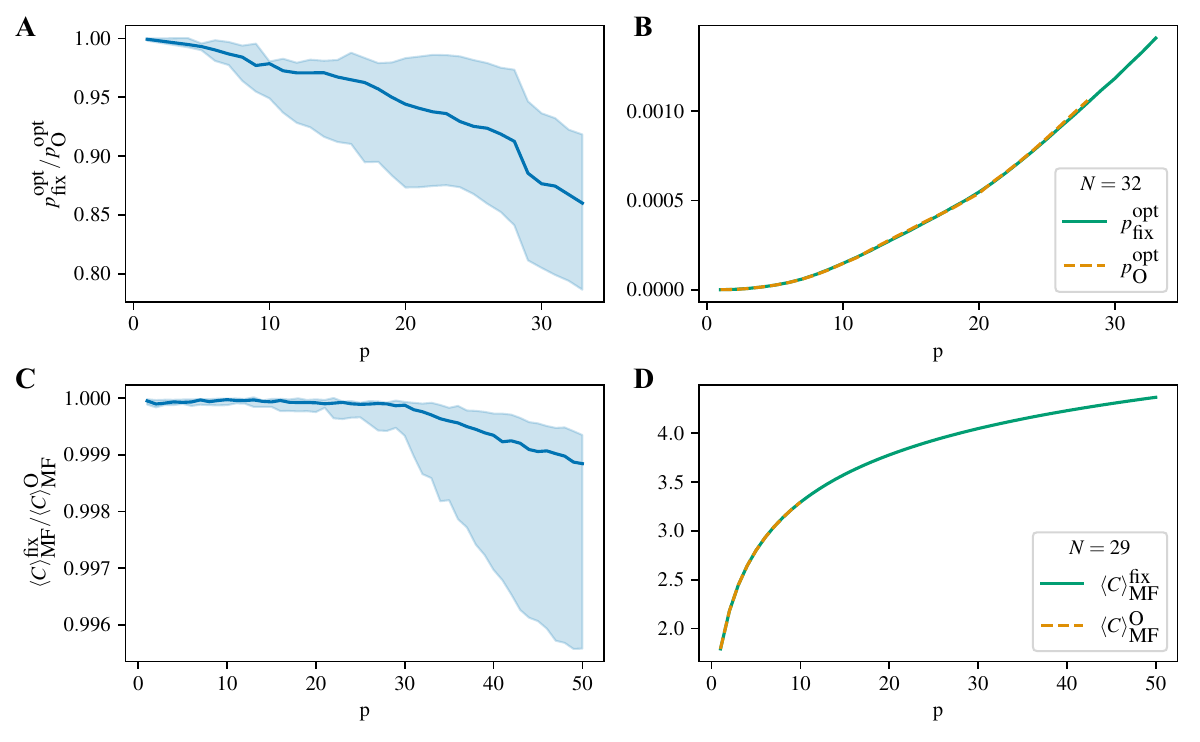}
    \caption{\textbf{QAOA with fixed parameters performs similarly to QAOA with directly optimized parameters} Comparison of QAOA performance with fixed ($\langle C\rangle_{\text{MF}}^{\text{fix}}$, $\popt_{\text{fix}}$) and optimized ($\langle C\rangle_{\text{MF}}^{\text{O}}$, $\popt_{\text{O}}$) parameters. \textbf{A,C,} Shaded area shows $95\%$ confidence interval. \textbf{B,D,} Despite the relative differences between performance with optimized and transferred parameters growing with $p$, we observe that for all cases considered, QAOA performance improves monotonically as expected. Since the absolute differences are small, the performance with fixed and directly optimized parameters is visually indistinguishable. 
    }
    \label{fig:fixed_vs_optimized}
\end{figure*}

To evaluate the quality of fixed parameters, we compare the QAOA performance with fixed parameters and with directly optimized parameters. Figure~\ref{fig:fixed_vs_optimized} presents the comparison. We observe that the two are very close for small $p$, with the ratio between the two growing for higher $p$. Specifically, for parameters optimized with respect to the expected merit factor $\langle C\rangle_{\text{MF}}$, the median difference in $\langle C\rangle_{\text{MF}}$ between QAOA evaluated with fixed and directly optimized parameters is less than $0.01\%$ for $p=50$, with the difference even lower for smaller $p$ (\ref{fig:fixed_vs_optimized}A). For parameters optimized with respect to $\popt$, the difference in $\popt$ is larger and growing with $p$ (\ref{fig:fixed_vs_optimized}C). Note that due to the exponentially small value of $\popt$, small absolute differences (including due to precision limitations) translate into large relative differences. Nonetheless, we observe good performance with fixed parameters at high $N$, as visualized for $N=32$ in Figure~\ref{fig:fixed_vs_optimized}D. QAOA performance with fixed parameters with respect to both figures of merit monotonically improves with $p$ for all values of $p$ considered. As the performance gap between fixed and optimized parameters grows with $p$, further improvements to fixed parameters are likely to yield even better scaling of QAOA TTS.

\begin{figure}
    \centering
    \includegraphics[width=\textwidth]{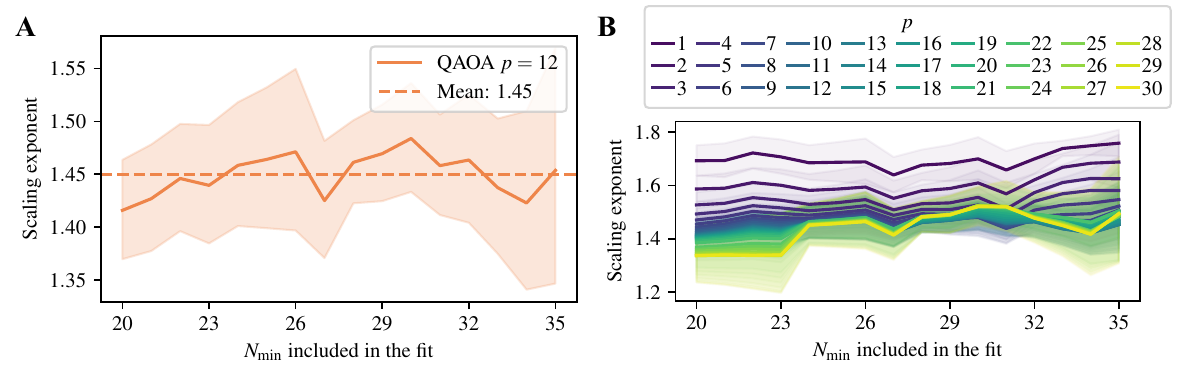}
    \caption{\textbf{Robustness of QAOA scaling} The scaling exponent for QAOA TTS for varying choices of $N_{\min}$ included in the range of values of $N$ for fit at \textbf{A,} $p=12$ and \textbf{B,} for varying $p$. If smaller $N_{\min}$ are included, the scaling exponent continues to improve with $p$, with the quality of fit decaying with $p$ (see Main text and Fig.~\ref{fig:power_law_fit_for_scaling_coeff}). For sufficiently high $N_{\min}$, the exponent does not improve beyond $p\approx 10$. At $p=12$, the scaling exponent is not sensitive to the choice of $N_{\min}$ (\textbf{A}).}
    \label{fig:QAOA_additional_scaling}
\end{figure}

\subsection{Scaling coefficient of QAOA TTS is not sensitive to the choice of $N_{\min}$}

In the Main text we motivate the choice of the cutoff $N_{\min}$ for the range of $N$s to be included in the fit by examining the stability of the quality of fit with varying $p$. We obtain $N_{\min}=28$ as the minimum value required to maintain a stable fit, and estimate QAOA scaling at $1.46^N$ at $p=12$. We now provide evidence that this value is not sensitive to the choice of $N_{\min}$. 

Fig.~\ref{fig:QAOA_additional_scaling} shows the scaling exponent for QAOA TTS for varying choices of the cutoff $N_{\min}$. Taking the average over exponents obtained when performing a fit with $20\leq N_{\min} \leq 35$ gives the estimated scaling of $1.45^N$ (Fig.~\ref{fig:QAOA_additional_scaling}a), which is slightly better than the one reported in the main text. As shown in Fig.~\ref{fig:QAOA_additional_scaling}b, for sufficiently large $p$ and small $N_{\min}$ the exponent changes as $N_{\min}$ is increased, indicating that a larger regime of $N$ must be explored to obtain a stable linear scaling.

\begin{figure*}
    \centering
    \includegraphics[width=0.5\textwidth]{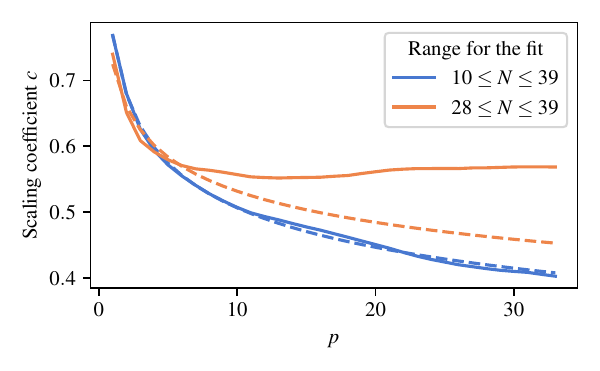}
    \caption{\textbf{QAOA scaling does not follow power law} The scaling coefficient $c$ in $\text{TTS}=\Theta(2^{cN})$ as a function of $p$. The blue line is included for illustration as including $N<28$ leads to low quality of the fit. When the quality of the fit is high ($28\leq N \leq 38$), the scaling coefficient does not follow a power law.}
    \label{fig:power_law_fit_for_scaling_coeff}
\end{figure*}

\subsection{Scaling coefficient of QAOA TTS does not follow power law}\label{sec:power_law}
One of the important findings of this work is that for LABS problem increasing QAOA depth $p$ beyond a certain small constant does not lead to better scaling. This puts the findings of this work in contrast to those of Refs.~\onlinecite{2208.06909,2212.01857,Akshay2022}. We observe that unlike in Ref.~\onlinecite{2208.06909}, the scaling coefficient $c$ in $\text{TTS}=\Theta(2^{cN})$ does not follow a power law.
This is shown in Figure~\ref{fig:power_law_fit_for_scaling_coeff}. For the coefficient $c$ to follow %
a power law, it must depend on $p$ as $c_1\times p^{c_2}$ for some constants $c_1$, $c_2$. When the cutoff is chosen to ensure good fit ($N=28$), we see clear deviation from a power law. We note that if we include smaller values of $N$ in the fit and ignore the low quality of the fit, we observe a clear power law scaling of $c\sim 0.77\times p^{-0.18}$. %

\subsection{Comparison of performance between QAOA and amplitude amplification}

\begin{figure}
    \centering
    \includegraphics[width=\textwidth]{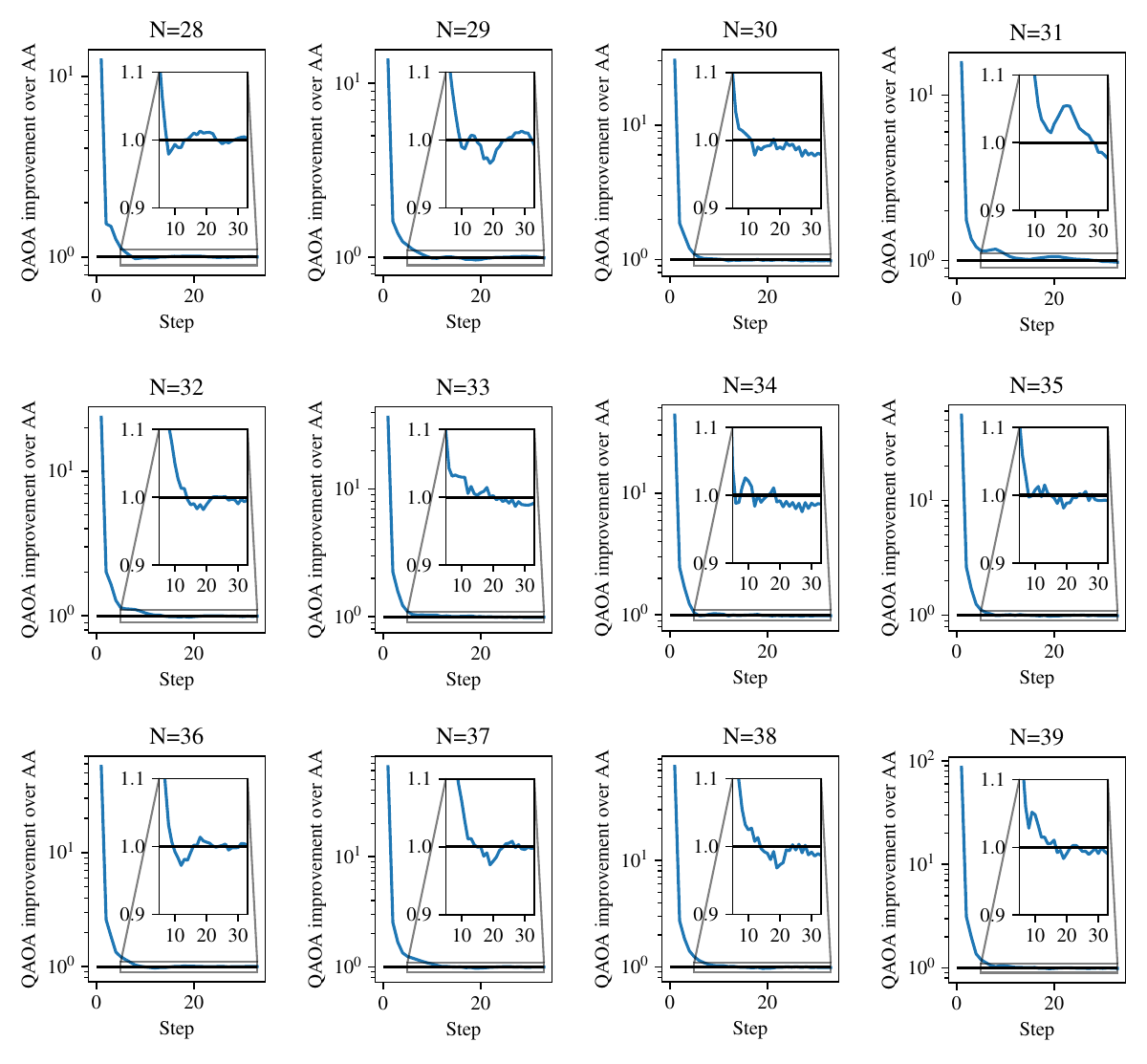}
    \caption{\textbf{Comparison of QAOA and amplitude amplification} The ratio between the gains in the probability of measuring the optimal solution from one step of QAOA and amplitude amplification (AA). QAOA provides orders of magnitude larger improvements for the first few steps, and then behaves approximately like amplitude amplification.}
    \label{fig:SI_QAOA_vs_aa}
\end{figure}

In this work we propose a strategy for using QAOA as a building block for algorithmic speedups by combining it with amplitude amplification (AA) or, more specifically, generalized minimum-finding as described in Appendix C of Ref.~\onlinecite{van_Apeldoorn_2020}. Specifically, we propose running the quantum minimum-finding algorithm with constant-depth QAOA as a subroutine. If the QAOA circuit prepares the state with $\popt$ overlap with the ground state of the LABS Hamiltonian, the minimum-finding algorithm would need apply the QAOA unitary and an oracle for computing the LABS cost function $O\left(1/\sqrt{\popt}\right)$ times to obtain a constant probability of measuring a bitstring corresponding to the optimal solution.

We note that our numerical results suggest that increasing the QAOA depth is always beneficial as compared to doing a smaller number of QAOA steps and increasing the number of amplitude amplification steps. This is visualized in Figure~\ref{fig:SI_QAOA_vs_aa}. As this figure shows, for sufficiently high $p$ the gain from a step of QAOA is very close to the gain from amplitude amplification, which provides an upper bound on the expected gain from one step of minimum-finding. At the same time one step of QAOA is much simpler to implement as it only requires the phase oracle and a series of one-qubit gates (mixer). We note that we do not provide exact gains from a step of the minimum-finding algorithm as they are non-trivial to estimate. However, asymptotically it is equivalent to amplitude amplification with an oracle that marks optimal solutions.

\newpage
\subsection{Proof of Theorem 1}

We now present the proof of Theorem 1 of the main text. Our technique is based on the generalized minimum-finding procedure outlined in Lemma 48 of Ref.~\onlinecite{van_Apeldoorn_2020}. For reader's convenience and completeness, we begin by restating our algorithm and the Theorem.

\begin{lemma}[Exponential Quantum Search, {Ref.~\onlinecite{Brassard_2002}}]
\label{lem:SI_eqsearch}
    Let $|\psi\rangle= U|0\rangle^{\otimes N}$ be a quantum state in a $2^N$-dimensional Hilbert space with computational basis elements indexed by $N$-bit bitstrings, and $m \colon \{0,1\}^{N} \to \{0,1\}$ be a marking function such that $\sum_{\{x | m(x) = 1\}} |\langle \psi | x \rangle|^2 \ge p$. There exists a quantum algorithm $\mathbf{EQSearch}(U,m,\delta)$ that outputs an element $x^*$ such that $m(x^*) = 1$ with probability at least \rev{$1-\delta$} using $O\left(\frac{1}{\sqrt{p}}\log\left(\frac{1}{\delta}\right)\right)$ applications of $U$ and $m$.
\end{lemma}

\begin{algorithm}[H]
\caption{QAOA Enhanced with Quantum Minimum-Finding}
\label{alg:SI_qaoa_aa}
\begin{algorithmic}
\Require Unitary $U_{\text{QAOA}}$ acting on $\mathbb{C}^{2^N}$ such that $\lvert\langle x^{*}|U_{\text{QAOA}}|0\rangle^{\otimes N}\rvert \geq 1/\sqrt{p_\text{opt}}$ for unknown $p_\text{opt}$, $V_{\text{LABS}}$ for computing $\mathcal{E}_{\text{sidelobe}}$
into a register, and $\delta \in (0, 1)$, positive number $M \le 2^N$, $C$ is the constant corresponding to the $O(\cdot)$ in Lemma~\ref{lem:SI_eqsearch}
\Ensure If $M$ is greater than $1/\sqrt{p_{\text{opt}}}$, output $x^{*}$ with $\geq 1 - \delta$ probability using $O(\log(1/\delta)M)$ calls to $U_{\text{QAOA}}$ and $V_{\text{LABS}}$ (and their inverses).
\State $x_{\text{res}}$ is set to an empty list.
\For{$i \gets 1$ \textbf{to} $\lceil \log(1/\delta) \rceil$}
\State $t \gets 0$; $s_0 \gets \infty$
\While {\text{number of calls to $U_{\text{QAOA}}$ \& $V_{\text{LABS}}$ is $<$ $3CMN$}}
\State $t \gets t + 1$
\State Define $m_t \colon \{0,1\}^{N} \to \{0,1\}$ such that $m_t(x) = 1$ if and only if $\mathcal{E}_{\text{sidelobe}}(x) < s_{t-1}$. Note that $m_t$ can be coherently evaluated using one query to $V_{\text{LABS}}$.
\State Set $s_t = \mathbf{EQSearch}(U_\text{QAOA},m_t,1/(6\cdot2^{N}))$.
\EndWhile
\State Append $s_t$ to $x_{\text{res}}$.
\EndFor
\State \text{Output minimum of $x_{\text{res}}$.}
\end{algorithmic}
\end{algorithm}

\begin{theorem}
Suppose a constant-depth QAOA circuit $U_{\text{QAOA}}$ prepares a state $|\psi\rangle = U_{\text{QAOA}}|0\rangle^{\otimes N}$ with $N \ge 3,$ such that we have $\lvert\langle x^{*}|\psi\rangle\rvert \geq 1/\sqrt{p_{\text{opt}}}$, where $|x^*\rangle$ encodes an optimal  solution to the $N$-bit LABS problem in a computational basis state, and we assume that $p_{\text{opt}} \ge 1/N$. Then, running Algorithm \ref{alg:SI_qaoa_aa} with parameters $M \ge 1/\sqrt{p_\text{opt}}$ and failure probability $\delta$, runs with a gate complexity of 
$O(\text{poly}(N)\log(1/\delta)M)$ and finds $x^{*}$ with probability at least $1 - \delta$.
\end{theorem}

\begin{proof}
We will first analyze the randomized Algorithm~\ref{alg:SI_qaoa_aa} as a \emph{Las Vegas} algorithm, i.e. we assume that the internal \textbf{while} loop is infinite.
In the upcoming analysis we will assume that every call to $\textbf{EQSearch}$ in the internal loop behaves as intended. Note that we choose the failure probability of each such call to be $1/(6\cdot 2^{N})$. The total number of calls cannot be more than $2^{N}$ and so by the union bound, every call succeeds except with probability at most $1/6$.

The algorithm generates a monotonically decreasing sequence of samples, uses $\mathbf{EQSearch}$ in each iteration to search for a sample that is strictly lesser than the previous one. Since the sequence of samples is strictly decreasing and there are $N$ possible distinct samples, the algorithm eventually returns the minimum. Since correctness is eventually guaranteed, we can simply bound the expected number of iterations before the sequence finds the minimum. If this expected number is $m$, we can run the internal loop $2m$ times to ensure that we find the minimum with probability at least $1/2$. Consequently, $2m\log\left(\frac{1}{\delta}\right)$ iterations suffice to ensure that we find the minimum with probability at least $1 - \delta$. It remains to show that the expected number of iterations before the minimum is found is at most $O\left(\frac{1}{\sqrt{p_{\text{opt}}}}\right)$.

For this argument, we define the following quantities. $(x_1 = x^*,x_2,\dots,x_n)$ is the list of bit-strings sorted in ascending order of the value of $\mathcal{E}_{\text{sidelobe}}$, and define $P(\mathcal{\xi}(X))$ to be the probability of event $\xi(X)$ when $X$ is a bitstring sampled by measuring $\ket{\psi}$ in the computational basis. Suppose in some iteration $t$, the sample returned by $\mathbf{EQSearch}$ is $s_{t}$. In the next iteration, $\mathbf{EQSearch}$ searches for an element with sidelobe energy less than $s_t$. The central observation is the following: for any $x_{k}$ where $k \in [N]$, the probability that some $s_t = x_k$ given that $t$ is the first iteration where $x_k$ appears in the list of obtained samples is given by $P(X = x_k)/P(X \le x_k)$. To see this, we observe as in Ref.~\onlinecite{van_Apeldoorn_2020} that:

\begin{align}
    \Pr(s_t = x_k | s_t \le x_k \wedge s_{t-1} > x_k) &= \frac{\Pr(s_t=x_k)}{\Pr(s_t \le x_k \wedge s_{t-1} > x_k)} \nonumber\\
    &= \sum_{x_l > x_k} \frac{\Pr(s_t=x_k \wedge s_{t-1} = x_{l-1})}{\Pr(s_t \le x_k \wedge s_{t-1} > x_k)} \nonumber\\
    &= \sum_{x_l > x_k} \frac{\Pr(s_t=x_k \wedge s_{t-1} = x_{l-1})\Pr(s_{t-1} = x_l)\Pr(s_t \le x_k \wedge s_{t-1} = x_l)}{\Pr(s_t \le x_k \wedge s_{t-1} > x_k)\Pr(s_{t-1} = x_l)\Pr(s_t \le x_k \wedge s_{t-1} = x_l)} \nonumber\\
    &= \sum_{x_l > x_k} \frac{\Pr(s_t=x_k \wedge s_{t-1} = x_{l-1})\Pr(s_t \le x_k \wedge s_{t-1} = x_l)}{\Pr(s_t \le x_k \wedge s_{t-1} > x_k)\Pr(s_t \le x_k \wedge s_{t-1} = x_l)} \nonumber\\
    &= \sum_{x_l > x_k} \frac{\Pr(s_t=x_k | s_{t-1} = x_{l-1})\Pr(s_{t-1} = x_l)\Pr(s_t \le x_k \wedge s_{t-1} = x_l)}{\Pr(s_t \le x_k | s_{t-1} = x_k)\Pr(s_{t-1} = x_l)\Pr(s_t \le x_k \wedge s_{t-1} > x_l)} \nonumber\\
    &= \sum_{x_l > x_k} \frac{P(X = x_k)\Pr(s_{t-1} = x_l)\Pr(s_t \le x_k \wedge s_{t-1} = x_l)}{P(X \le x_k)\Pr(s_{t-1} = x_l)\Pr(s_t \le x_k \wedge s_{t-1} > x_l)} \nonumber\\
    &= \frac{P(X = x_k)}{P(X \le x_k)}.
\end{align}
Notice that since a value can be sampled at most once, and the minimum is obtained within $n$ steps, the probability that a given value $x_k$ occurs in the list of observed samples is $P(X = x_k)/P(X \le x_k)$. 

To bound the expected number of queries before the minimum ($x_1$) is found by the algorithm, we associate with each bitstring $x_l \in \{x_n,x_{n-1},\dots,x_2\}$ the probability that it is an obtained sample in some iteration, and the cost of performing the corresponding search for an element less than $x_l$. The cost of the search, for our chosen parameters is at most $\frac{C\log(6\cdot2^{N})}{\sqrt{P(X < x_l)}} \le \frac{2CN}{\sqrt{P(X < x_l)}}$, where the last inequality holds whenever $N \ge 3$.

The total expected number of queries before the minimum $x_1$ is found, is therefore given by

\begin{align}
& \sum_{i = 2}^{\rev{n}} 2CN \cdot \frac{P(X = x_i)}{P(X \le x_i)}\sqrt{\frac{1}{P(X < x_i)}} \le \sum_{i = 2}^{\rev{n}} \frac{2CN}{P(X \le x_i)}\sqrt{\frac{1}{P(X < x_i)}} \le 2CN\int_{p_{\text{opt}}}^1 r^{-3/2} \,dr \le \frac{CN}{\sqrt{p_{\text{opt}}}}.
\end{align}
\rev{Note that the second to last inequality follows from arguments made in Ref.~\onlinecite{van_Apeldoorn_2020}, where they used the same inequality.}

If we run the inner loop of Algorithm~\ref{alg:SI_qaoa_aa} more than $3$ times the expected number of queries required to find $x^*$, as prescribed if $M \ge 1/\sqrt{p_\text{opt}}$, we fail to find $x^*$ with probability at most $1/3$ by the Markov inequality as long as no query to 
$\textbf{EQSearch}$ fails. Additionally, by the earlier discussion, a query in the internal loop fails with probability at most $1/6$. Therefore, by a union bound, each application of the inner loop finds $x^{*}$ with probability of at least $1/2$. Repeating the inner loop $\log(1/\delta)$ times ensures that $x^{*}$ is found with probability at least $1 - \delta$ (if not, the inner loop has to fail to find $x^*$ in $\log(1/\delta)$ independent trials).
\end{proof}

\subsection{Details of the classical solver scaling}

The scaling of the commercial branch-and-bound solvers is presented in Figure~\ref{fig:CPLEX_and_Gurobi}. For each solver, we run it with 100 random seeds for $N\leq 32$ and 10 random seeds for $N>32$ and report the mean. The minimum $N$ to include in the fit was chosen to maximize the quality of fit. We set the Gurobi parameters as follows: \texttt{Cuts=0}, \texttt{Heuristics=0}. For the other parameters in Gurobi and CPLEX the defaults are used. We observe that the performance of the two commercial solvers considered is within a $95\%$ confidence interval of each other, with the TTO scaling matching that reported in Ref.~\onlinecite{Packebusch2016}.

We report complete results for the Memetic Tabu solver in Figure~\ref{fig:tabu}. The scaling is obtained by extrapolating the number of cost function evaluations made by the Memetic Tabu algorithm at different lengths. This quantity is fixed over repeated seeds for a given seed and length, unlike the execution time that may fluctuate depending on the runtime environment. The fluctuation in running time is much larger for Memetic Tabu as compared to branch-and-bound due to lower absolute value of the runtime ($<1$ sec for $N=40$)
The time to evaluate the cost function on a sequence of length $N$ scales only as $N^2$, which does not produce a consistent effect on runtime scaling at small lengths. The TTS scaling is therefore essentially the same as the scaling of the number of function evaluations,\cite{Bokovi2017} and we choose to report the latter, more stable quantity.
The seeds chosen for the runs are a contiguous block of 50 integers chosen from a random starting point. The Memetic Tabu solver is run in single-threaded mode for all our experiments, to avoid the overestimation of cost function evaluations arising from race conditions between exploration and termination checks.

\begin{figure}
    \centering
    \includegraphics[width=\textwidth]{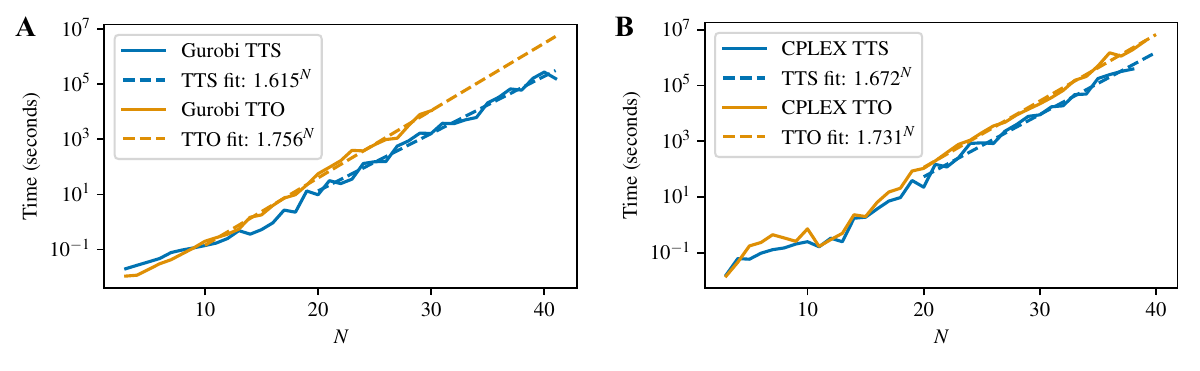}
    \caption{\textbf{Scaling of exact solvers} Time-to-solution (TTS) and time to obtain a certificate of optimality (time-to-optimality or TTO) of \textbf{A,} Gurobi and \textbf{B,} CPLEX. For Gurobi, the $95\%$ confidence interval (CI) for TTS is $(1.571,1.659)$ and for TTO is $(1.721,1.792)$. For CPLEX, the $95\%$ CI for TTS is $(1.609,1.737)$ and for TTO is $(1.693,1.770)$. %
    }
    \label{fig:CPLEX_and_Gurobi}
\end{figure}%

\begin{figure}
    \centering
    \includegraphics[width=0.48\textwidth]{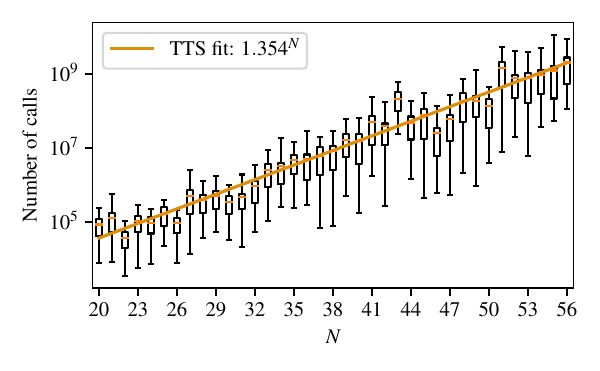}
    \caption{\textbf{Scaling of heuristic solver}Box plots showing the run time of Memetic Tabu. The scaling fit is reported with respect to the mean. The $95\%$ confidence interval (CI) is $(1.325,1.383)$. The whiskers are minimum and maximum, box is showing quartiles and the horizontal line in the box is the mean. 
    }
    \label{fig:tabu}
\end{figure}

\section{Experiments on trapped-ion systems}

\subsection{Experimental system}
The experiments in this work were performed on Quantinuum H1 and H2 platforms~\cite{Pino2021,Moses2023RaceTrack}.
The system design is based on the QCCD architecture with multiple separate gate zones.
Each gate zone is used to perform operations on an arbitrary pair of two qubits at a time,
suppressing crosstalk and maintaining high fidelity.  
The hyperfine approximate clock states of $^{171}$Yb$^{+}$ in the $^2 S_{1/2}$ state are used to encode qubit information. Namely, $\ket{0} \equiv \ket{F=0, m_f=0}$ and $\ket{1} \equiv \ket{F=1, m_f=0}$. After loading, the qubits can be prepared in $\ket{0}$ via optical pumping~\cite{Pino2021,Moses2023RaceTrack}.

The systems have all-to-all connectivity with two-qubit gates between pairs of qubits implemented by ion transport, which brings the pairs into the same gate zone. To implement two-qubit gates, a phase-sensitive M\o lmer-S\o rensen (MS) gate sandwiched between single-qubit wrapper pulses is used. It in turn gives the $\zgate\zgate$ gate $R_{\zgate\zgate}(\gamma)=\exp (-i\gamma \zgate\zgate/2)$, where the rotation angle can be precisely controlled by changing the parameters in the MS gate.\cite{Pino2021,Moses2023RaceTrack} Both the H1-1 and H2 systems used in this work have a typical average  two-qubit infidelity of $2 \times 10^{-3}$, with single qubit infidelity two orders of magnitude smaller. The qubit state can be read out via state-dependent resonance fluorescence measurement. Note that mid-circuit measurement and reset can be implemented while causing a small crosstalk error due to the stray light from the measurement and reset laser beams.

\subsection{Circuit compilation and optimization}
We now present the circuit compilation procedure for the experiments on the trapped-ion quantum processor. The H-series devices used in this work have at most five gate zones~\cite{Moses2023RaceTrack}, i.e. at most five two-qubit gates can be executed in parallel. This implies that optimizing the circuit for full parallelism may result in diminishing returns past five parallel gates. In any case, the highest error operation on the devices is the two-qubit gate, so the primary limiting factor for achieving high-fidelity results is the two-qubit gate count. %
Thus, we chose to first optimize the two-qubit gate count. In addition, the cost operator is the composition of diagonal gates, and thus we are free to apply them in any order. We optimize the order to maximize the number of gate cancellations. A similar approach has been used in Ref.~\onlinecite{2211.11287} for devices with nearest-neighbor connectivity.

We start by decomposing the four-body terms $R_{\zgate\zgate\zgate\zgate}(\gamma)$ into four $\cxgate$s and a single $R_{\zgate\zgate}(\gamma)$, where $R_{\zgate\zgate\zgate\zgate}(\gamma)$ and $R_{\zgate\zgate}(\gamma)$ denote evolution under $\zgate\zgate\zgate\zgate$ and $\zgate\zgate$ coupling with angle $\gamma$, respectively. 
Note that the $R_{\zgate\zgate}(\gamma) = e^{-i \frac{\gamma}{2} \zgate\zgate}$ is the native gate for the Quantinuum H-series trapped-ion processors. %

\begin{figure}[htb!]
    \includegraphics[width=0.5\linewidth]{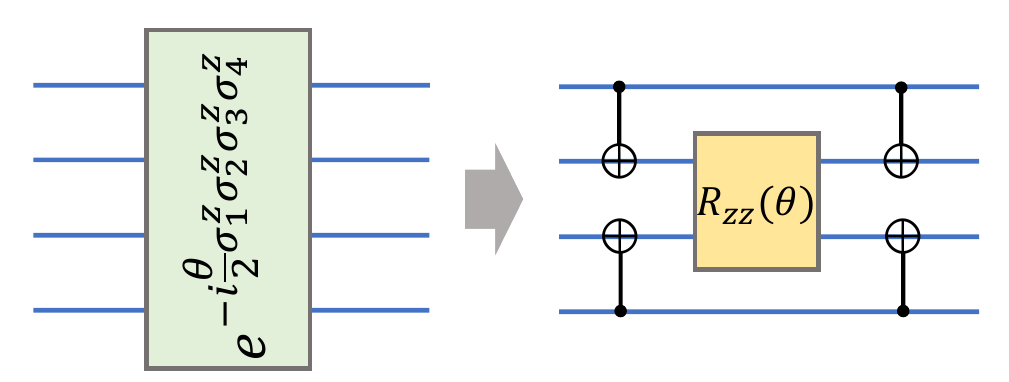}
    \caption{\textbf{Compilation of four-body terms} Decomposition of four-body interaction terms into a two-body $R_{\zgate\zgate}$  gate and four $\cxgate$'s. Figure reproduced from the main text.}
    \label{fig:rzzzz_decomp}
\end{figure}

\begin{figure}[htb!]
    \includegraphics[width=0.5\linewidth]{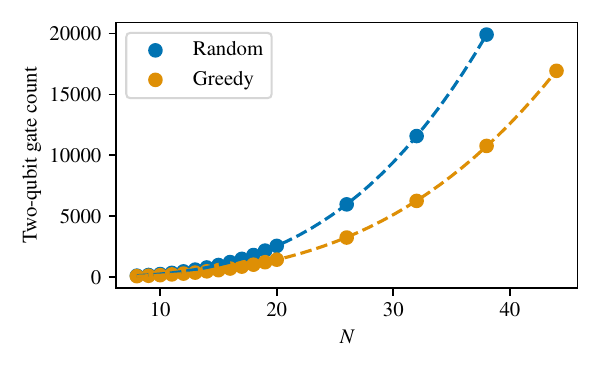}
    \caption{\textbf{Greedy optimization reduces the gate count} Comparison of two-qubit gate count (number of $R_{\zgate\zgate}(\gamma)$ + $R_{\zgate\zgate}(\pi/2)$) of QAOA circuit at $p=1$ with random term ordering and with the greedy optimization. The ``random'' line is the average over 20 random orderings of the two- and four-body terms. Both random and greedy are further optimized by tket~\cite{Sivarajah2020},
    and the resulting gate count is plotted. Cubic fit line added to guide the eye.}
    \label{fig:greedy_compare}
\end{figure}

The goal of the first step of the compilation procedure is to schedule the $R_{\zgate\zgate\zgate\zgate}(\gamma)$ gates corresponding to four-body terms in a way that greedly cancels as many $\cxgate$s as possible (see Figure \ref{fig:rzzzz_decomp} for the decomposition). The second step deals with the two-body terms and attempts to schedule each $R_{\zgate\zgate}(\gamma)$ gate near a four-body term where one of the $\cxgate$ acts on the same qubits as the two-body term. This is to leverage two-qubit gate resynthesis implemented in tket.\cite{Sivarajah2020}  These steps are described in detail in Algorithm~\ref{alg:cost_operator_opt}. Then, the resulting circuit is passed to tket circuit optimizer to transpile the circuit into the H-series native two-qubit gates: $R_{\zgate\zgate}(\gamma)$ and $R_{\zgate\zgate}(\pi/2)$. The preliminary step of greedly optimizing the layout of the interactions  reduced the two-qubit gate count by $1.7$ times on average compared to tket alone. The improvement in two-qubit gate count from the greedy $\cnotgate$ cancellation is shown in Fig.~\ref{fig:greedy_compare}.
Lastly, we compared a variety of circuit optimizers and found that  tket resulted in the largest gate-count reduction.

\begin{algorithm}[H]
\caption{Greedy cost-operator circuit optimization}
\label{alg:cost_operator_opt}
\begin{algorithmic}
\Require \\
\quad List  of four-tuples of  $(i, j, k, l)$, with $i < j < k < l$, where $(i, j)$ and $(k, l)$, respectively correspond to the qubits of top two $\cnotgate_{ij}$ and bottom two $\cnotgate_{lk}$ in decomposition of $R_{\zgate\zgate\zgate\zgate}$ presented in  Figure \ref{fig:rzzzz_decomp}. Note the indices corresponding to the control and target, respectively, are reversed for $(i, j)$ and $(k, l)$\\
\quad List of two-tuples $(i, j)$, with $i < j$, indicating which qubits to apply each $R_{\zgate\zgate}$ rotation to.
\Ensure Output a single list that contains all of the terms (both four- and two-body) in the order in which they should be applied, according to the greedily-optimized circuit.
\\
\State $\mathrm{circuit} \gets \text{empty list}$

\For{each collection of four body terms $(i, j, k, l)$ grouped by locality $d : = j -i~(= l -k~\text{for LABS})$}
\State current $\gets$ uniformly randomly sample (and remove) a term $(i, j, k, l)$ from the collection of terms of locality $d$.
\State add current to the circuit list
\State tops $\gets$ list initialized with tuple $(i, j)$
\State bottoms $\gets$ list initialized with tuple $(k, l)$
\While {there are still more terms in the collection}
\For {each term $(r, s ,t, v)$ in the  collection}
\If{$(r, s) \in$ tops or $(t, v) \in $ bottoms}
\State {Assign the term a score of $+1$}
\Else 
\State {Assign the term a score of $-1$.}
\EndIf
\If {$\exists\, m\:|\:(m, r) \in$ bottoms or $\exists\, a\:|\:(a, t) \in$ tops}  
\State Subtract $1$ from terms score. // This implies that inserting this term to the circuit would mean that there is some  $\cxgate_{mr}$ or $\cxgate_{ta}$ currently in the circuit that can never be cancelled.
\EndIf
\EndFor
\State current $\gets$ term $(a, b, c, d)$ in collection with highest score
\State add current to the circuit list
\State Add $(a, b)$ to tops (if not already in) and $(c, d)$ to bottoms (if not already in)
\EndWhile
\EndFor

\For{each two body term $(i, j)$}
\For{each four body  term $(r, s ,t ,v)$ in circuit}
\State {insert $(i, j)$ after $(r, s ,t ,v)$ in circuit if $(i,j) =(r,s)$ or $(i,j)=(t,v)$ and break inner loop}
\EndFor
\If{$(i, j)$ was not inserted into circuit}
\State {Add  $(i, j)$ to end of circuit}
\EndIf
\EndFor

\State output circuit
\end{algorithmic}
\end{algorithm}

We note that there exists an alternative proposal \cite{Sanders2020} for reducing the cost of implementing the LABS cost operator. In this approach, the phase operator is replaced by the evolution under the Hamiltonian corresponding to
    \begin{equation}
    \mathcal{E}_{\text{sidelobe}}(\boldsymbol{s}) = \sum_{k=1}^{N-1} \lvert \mathcal{A}_k(\boldsymbol{s})\rvert,
\end{equation}
where 
\begin{equation}
    \mathcal{A}_k(\boldsymbol{s}) = \sum_{i=1}^{N-k}s_is_{i+k}.
\end{equation}
This is in contrast to the approach mentioned in the main text:
\begin{equation}\label{eq:SI_labs_en}
    \mathcal{E}_{\text{sidelobe}}(\boldsymbol{s}) = \sum_{k=1}^{N-1} \mathcal{A}_k^2(\boldsymbol{s}).
\end{equation}
The Hamiltonian corresponding to the absolute value of the autocorrelations has the same ground state space and energy levels as the one in Equation \ref{eq:SI_labs_en}. While this reduces the asymptotic complexity for computing the energy to $\Theta(N^2)$, it requires quantum arithmetic, putting it beyond the capability of current hardware. Thus, we focus on optimizing the cost operator corresponding to Equation \ref{eq:SI_labs_en}. Further gate count reductions may be possible by fixing some of the variables and applying the techniques of Ref.~\onlinecite{Ayanzadeh2023}, though doing so is outside of the scope of this work.

\subsection{Summary of the error-detection scheme}

\begin{figure}
    \centering
    \includegraphics[width=0.6\textwidth]{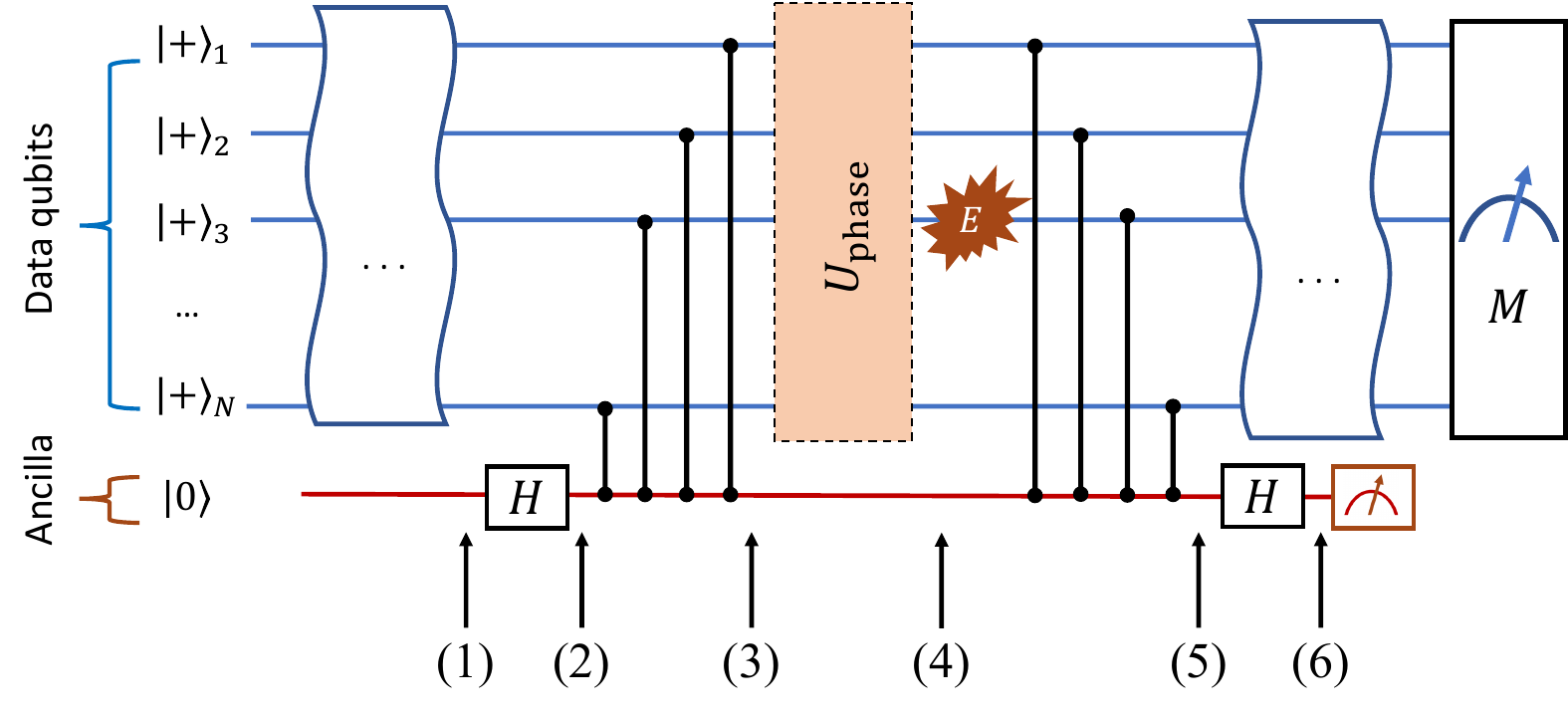}
    \caption{\textbf{Overview of one step of the error-detection scheme.} A part of the circuit $U_{\text{phase}}$ is ``sandwiched'' between two parity checks. Any error on data qubits that does not commute with the check is guaranteed to be detected.}
    \label{fig:one_parity_check}
\end{figure}

We now briefly summarize the error-detection scheme. The proof that our scheme is capable of detecting an arbitrary single-qubit error in the phase operator circuit under the assumption of noiselessly implemented parity checks is a special case of Ref.~\onlinecite[Theorem 1]{Gonzales2023}. We include a brief discussion of the scheme here for completeness and refer interested readers to Ref.~\onlinecite{Gonzales2023} for a detailed discussion. We first consider the case of when just one parity is checked (either $\xgate$ or $\zgate$), and then show how the analysis generalizes to the case when both parities are checked simultaneously. Figure~\ref{fig:one_parity_check} presents the overview of the circuit. Let the state of the data qubits before the check be $\rho_{\text{init}}$ and assume the ancilla is perfect and initialized to the pure state $\ket{0}$. The $\zgate$ parity check is given by the operator $C_{\zgate} = \frac{1}{2}\left(\idgate \otimes \ket{0}\bra{0} + \zgate^{\otimes N}\otimes \ket{1}\bra{1}\right)$.
Denote the part of the phase operator sandwiched between the checks as $U_{\text{phase}}$. Note that $U_{\text{phase}}$ can encompass the full phase operator or only a part of it. At (1), the state is $\rho_{(1)} = \rho_{\text{init}}\otimes \ket{0}\bra{0}$. At (2), the state is $\rho_{(2)} = \rho_{\text{init}}\otimes \frac{\ket{0}\bra{0}+\ket{1}\bra{1}}{2}$. After the first check is applied, at (3) the state becomes 
\begin{equation}
\rho_{(3)} = \frac{1}{2}\left(\rho_{\text{init}}\otimes \ket{0}\bra{0} + \zgate^{\otimes N}\rho_{\text{init}}\zgate^{\otimes N}\otimes\ket{1}\bra{1}\right).    
\end{equation}
If no error during the application of $U_{\text{phase}}$ occurs, then the state at point (4) is 
\begin{equation}
\rho_{(4)}^{\text{no error}} = \frac{1}{2}\left(U_{\text{phase}}\rho_{\text{init}}U_{\text{phase}}^{\dagger}\otimes \ket{0}\bra{0} + U_{\text{phase}}\zgate^{\otimes N}\rho_{\text{init}}\zgate^{\otimes N}U_{\text{phase}}^{\dagger}\otimes\ket{1}\bra{1}\right),    
\end{equation}
and 
\begin{align}
\rho_{(5)}^{\text{no error}} & = \frac{1}{2}\left(U_{\text{phase}}\rho_{\text{init}}U_{\text{phase}}^{\dagger}\otimes \ket{0}\bra{0} + \zgate^{\otimes N}U_{\text{phase}}\zgate^{\otimes N}\rho_{\text{init}}\zgate^{\otimes N}U_{\text{phase}}^{\dagger}\zgate^{\otimes N}\otimes\ket{1}\bra{1}\right) \\
 & = \frac{1}{2}\left(U_{\text{phase}}\rho_{\text{init}}U_{\text{phase}}^{\dagger}\otimes \ket{0}\bra{0} + U_{\text{phase}}\rho_{\text{init}}U_{\text{phase}}^{\dagger}\otimes\ket{1}\bra{1}\right) \\
 & = U_{\text{phase}}\rho_{\text{init}}U_{\text{phase}}^{\dagger}\otimes\frac{\ket{0}\bra{0}+\ket{1}\bra{1}}{2},
\end{align}
so the final state is
\begin{equation}
\rho_{(6)}^{\text{no error}}  = U_{\text{phase}}\rho_{\text{init}}U_{\text{phase}}^{\dagger}\otimes\ket{0}\bra{0}.
\end{equation}
Therefore if no error occurred, measuring the ancillary qubit will always give the measurement outcome $0$. 

If a single-qubit Pauli error $E$ occurred during the execution of $U_{\text{phase}}$, the state at (4) becomes
\begin{align}
\rho_{(4)} = \frac{1}{2}\left(EU_{\text{phase}}\rho_{\text{init}}U_{\text{phase}}^{\dagger}E^{\dagger}\otimes \ket{0}\bra{0} + EU_{\text{phase}}\zgate^{\otimes N}\rho_{\text{init}}\zgate^{\otimes N}U_{\text{phase}}^{\dagger}E^{\dagger}\otimes\ket{1}\bra{1}\right),
\end{align}
and 
\begin{align}
\rho_{(5)} = \frac{1}{2}\left(EU_{\text{phase}}\rho_{\text{init}}U_{\text{phase}}^{\dagger}E^{\dagger}\otimes \ket{0}\bra{0} + \zgate^{\otimes N}EU_{\text{phase}}\zgate^{\otimes N}\rho_{\text{init}}\zgate^{\otimes N}U_{\text{phase}}^{\dagger}E^{\dagger}\zgate^{\otimes N}\otimes\ket{1}\bra{1}\right).
\end{align}
If $E$ is a Pauli, it can either commute or anti-commute with the check $\zgate^{\otimes N}$. If it anti-commutes, then 
\begin{align}
\rho_{(6)} & = \frac{1}{2}\left(EU_{\text{phase}}\rho_{\text{init}}U_{\text{phase}}^{\dagger}E^{\dagger}\otimes \ket{0}\bra{0} - E\zgate^{\otimes N}U_{\text{phase}}\zgate^{\otimes N}\rho_{\text{init}}\zgate^{\otimes N}U_{\text{phase}}^{\dagger}\zgate^{\otimes N}E^{\dagger}\otimes\ket{1}\bra{1}\right) \\
& = EU_{\text{phase}}\rho_{\text{init}}U_{\text{phase}}^{\dagger}E^{\dagger}\otimes\frac{\ket{0}\bra{0}-\ket{1}\bra{1}}{2},
\end{align}
and the final state 
\begin{equation}
\rho_{(6)}  = U_{\text{phase}}\rho_{\text{init}}U_{\text{phase}}^{\dagger}\otimes\ket{1}\bra{1}.
\end{equation}
Then measuring the ancillary qubit will return 1, and the error will be detected. A Pauli error that commutes with the check, however, will go undetected. 

When both $\xgate^{\otimes N}$ and $\zgate^{\otimes N}$ parities are checked, no odd-weight Pauli commutes with both checks. Therefore, any odd-weight Pauli error will be detected by the implemented scheme. If the checks are noiseless, increasing the frequency of checks by reducing the size of the circuit $U_{\text{phase}}$ between the checks can only improve the final fidelity. However, in practice the checks are noisy, introducing a trade-off between the increase in the errors detected and the errors introduced by the checks themselves.

\begin{figure}[t]
\includegraphics[width=\linewidth]{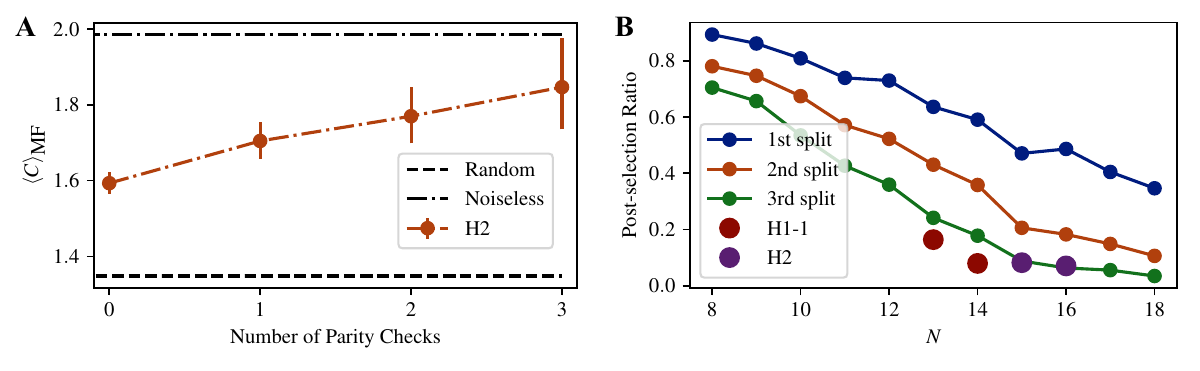}
    \caption{\textbf{Using more parity checks gives better performance} \textbf{A,} Expected merit factor as a function of the number of error checks used when post-selecting the data qubit results for the $N=15$ experiment on the H2 hardware system with 5000 repetitions. Here $m=3$ is the number of $\zgate$- and $\xgate$-parity checks.  The error bars become larger as we discard more shots due to more errors being detected with more parity checks. The first data point corresponds to the case when we keep all the results from measurement on the data qubits and do not do post-selection. \textbf{B,} Post-selection ratio when splitting the phase operator into three parts and performing parity syndrome measurement at the end of each part. The curves show the simulation results using realistic parameters of the Quantinuum H2 trapped-ion device. The measured results on both H1 and H2 hardware with three checks are in good agreement with simulations.}
    \label{fig:MF_post_selection_ratio}
\end{figure}

\subsection{Performance of the error-detection scheme}

In this section, we provide details on the implementation of the proposed error-detection scheme. For a circuit with $m$ checks, we separate the phase operator into $m$ parts, where each part has approximately the same number of two-qubit gates.
We observe that increasing the frequency of the parity checks up to $m=3$ improves the QAOA performance. Fig.~\ref{fig:MF_post_selection_ratio}A shows improvements in expected merit factor after post-selection. More measurements increases the ratio of samples with detected errors, increasing the overhead of the error-detection scheme in terms of the number of repetitions. Fig.~\ref{fig:MF_post_selection_ratio}B shows how this overhead increases with $N$ by plotting the decay of the ratio of samples with no error detected to all samples (``post-selection ratio''). This ratio drops to below $10\%$ at $N\geq 15$ and $m=3$. To trade off the number of repetitions required against the fidelity of the final result, we set $m=3$ in our experiments.
We further note that in our experiments, the post-selection typically keeps the bitstring with the highest merit factor sampled from experiments, as shown in Fig.~\ref{fig:best_sampled_MF}. Note that in a practical optimization setting, the best of all bitstrings corresponding to valid solutions would be chosen as the output.

\begin{figure}
\includegraphics[width=0.5\linewidth]{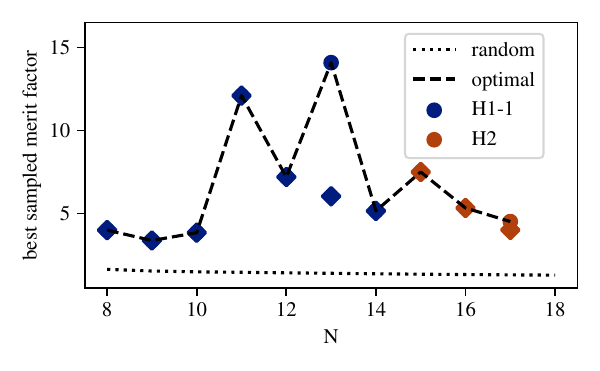}
    \caption{\textbf{Post-selection keeps the best solution} Experimentally sampled bitstring that has the highest merit factor for different $N$. The number of shots taken is 2000 for $N=8,9$ and 5000 for other $N$. Circles (diamonds) are the data without (with) post-selection. They have exact same values except for $N=13$ and $N=17$. The best sampled bitstrings before post-selection have the same merit factor as true optimal bitstrings for all instances studied here.}
    \label{fig:best_sampled_MF}
\end{figure}

An important benefit of our error-detection scheme is reduced time to a high-quality sample, i.e. a sample with no errors detected. The time improvement comes from the ability to stop the execution when a mid-circuit measurement detects an error. This capability is particularly relevant to trapped-ion systems with relatively low clock speeds and very long coherence times enabling such classical feedback. Although available hardware supports this feature, we do not implement the early stopping in our hardware experiments. The time savings provided below are estimates.

We denote the probability that no detectable error occurs during part $i$ as $p_i$. For the case without any mid-circuit syndrome measurement, the average time needed to reach a measurement result with a high merit factor, i.e., no parity error detected for all the check measurements, is given by
\begin{equation}
    \bar{t}_{1} = t_0/(\prod_i^m p_i),
\end{equation}
where $t_0$ is the total circuit time. With mid-circuit check and feed-forward discard of the remaining circuit conditioned on the check result, the average time to get a bitstring result for which the merit factor has a high value reads as 
\begin{equation}
    \bar{t}_2 = t_0/(\prod_i^m p_i) \times (\sum_i^{m-1} (\prod_{j=0}^{i-1} p_{j}(1-p_i) \frac{i}{m}) + \prod_k^{m-1} p_k),
\end{equation}
with $p_0=1$. Here we neglect the gate time between data qubits and ancillary qubits. The comparison between the two is shown in Fig.~\ref{fig:AvgTime2goodShot}, indicating that our error check method would reduce the time to get a bitstring with a high merit factor.

\begin{figure}
    \includegraphics[width=0.5\linewidth]{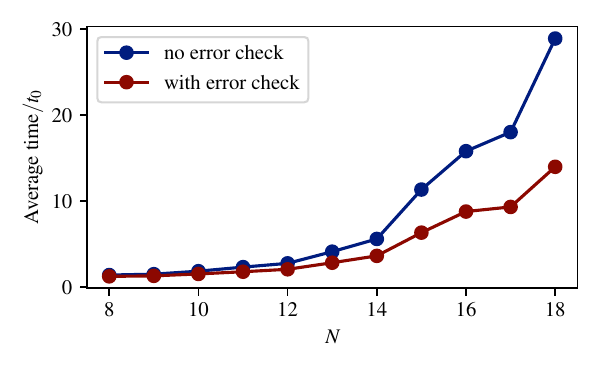}
    \caption{\textbf{Post-selection reduces time to obtain a good bitstring} Simulation of average time (normalized by $t_0$) to a bitstring without detectable parity check errors. The number of splits $m$ is set to $3$. The simulations are performed using realistic parameters of the Quantinuum H2 trapped-ion device. }
    \label{fig:AvgTime2goodShot}
\end{figure}

\end{document}